\documentclass{article} 
\usepackage{algorithmic}
\usepackage{algorithm}
\usepackage{times} 
\usepackage{amsfonts}
\usepackage{amsmath} 
\usepackage{latexsym} 
\usepackage{epsfig}
\usepackage{latexsym} 
\usepackage{verbatim}
\begin{document}
 
\newtheorem{theorem}{Theorem}
\newtheorem{corollary}[theorem]{Corollary}
\newtheorem{prop}[theorem]{Proposition} 
\newtheorem{problem}[theorem]{Problem}
\newtheorem{lemma}[theorem]{Lemma} 
\newtheorem{remark}[theorem]{Remark}
\newtheorem{observation}[theorem]{Observation}
\newtheorem{defin}{Definition} 
\newtheorem{example}{Example}
\newtheorem{conj}{Conjecture} 
\newenvironment{proof}{{\bf Proof:}}{\hfill$\Box$\linebreak\vskip2mm} 
\newcommand{\PR}{\noindent {\bf Proof:\ }} 
\def\EPR{\hfill $\Box$\linebreak\vskip3mm} 
 \def\EEX{\hfill $\diamond$}

\def\Pol{{\sf Pol}} 
\def\mPol{{\sf MPol}} 
\def\Polo{{\sf Pol}_1} 
\def\PPol{{\sf pPol\;}} 
\def\Inv{{\sf Inv}}
\def\mInv{{\sf MInv}} 
\def\Clo{{\sf Clo}\;} 
\def\Con{{\sf Con}} 
\def\concom{{\sf Concom}\;} 
\def\End{{\sf End}\;}
\def\Sub{{\sf Sub}\;} 
\def\Im{{\sf Im}} 
\def\Ker{{\sf Ker}\;} 
\def\H{{\sf H}}
\def\S{{\sf S}} 
\def\D{{\sf P}} 
\def\I{{\sf I}} 
\def\Var{{\sf var}} 
\def\PVar{{\sf pvar}} 
\def\fin#1{{#1}_{\rm fin}}
\def\P{{\sf P}} 
\def\Pfin{{\sf P_{\rm fin}}} 
\def\Id{{\sf Id}}
\def\R{{\rm R}} 
\def\F{{\rm F}} 
\def\Term{{\sf Term}}
\def\var#1{{\sf var}(#1)} 
\def\Sg#1{{\sf Sg}(#1)} 
\def\Sgg#1#2{{\sf Sg}_{#1}(#2)} 
\def\Cg#1{{\sf Cg}(#1)}
\def\Cgg#1#2{{\sf Cg}_{#1}(#2)}
\def\Cen{{\sf Cen}}
\def\tol{{\sf tol}} 
\def\lnk{{\sf lk}} 
\def\rbcomp#1{{\sf rbcomp}(#1)}
  
\let\cd=\cdot 
\let\eq=\equiv 
\let\op=\oplus 
\let\omn=\ominus
\let\meet=\wedge 
\let\join=\vee 
\let\tm=\times
\def\ldiv{\mathbin{\backslash}} 
\def\rdiv{\mathbin/}
  
\def\typ{{\sf typ}} 
\def\zz{{\un 0}} 
\def\zo{{\un 1}}
\def\one{{\bf1}} 
\def\two{{\bf2}} 
\def\three{{\bf3}}
\def\four{{\bf4}} 
\def\five{{\bf5}}
\def\pq#1{(\zz_{#1},\mu_{#1})}
  
\let\wh=\widehat 
\def\ox{\ov x} 
\def\oy{\ov y} 
\def\oz{\ov z}
\def\of{\ov f} 
\def\oa{\ov a} 
\def\ob{\ov b} 
\def\oc{\ov c}
\def\od{\ov d} 
\def\oob{\ov{\ov b}} 
\def\rx{{\rm x}}
\def\rf{{\rm f}} 
\def\rrm{{\rm m}} 
\let\un=\underline
\let\ov=\overline 
\let\cc=\circ 
\let\rb=\diamond 
\def\ta{{\tilde a}} 
\def\tz{{\tilde z}}
\let\td=\tilde
\let\dg=\dagger
\let\ddg=\ddagger
  
  
\def\zZ{{\mathbb Z}} 
\def\B{{\mathcal B}} 
\def\P{{\mathcal P}}
\def\zL{{\mathbb L}} 
\def\zD{{\mathbb D}}
 \def\zE{{\mathbb E}}
\def\zG{{\mathbb G}} 
\def\zA{{\mathbb A}} 
\def\zB{{\mathbb B}}
\def\zC{{\mathbb C}} 
\def\zM{{\mathbb M}} 
\def\zR{{\mathbb R}}
\def\zS{{\mathbb S}} 
\def\zT{{\mathbb T}} 
\def\zN{{\mathbb N}}
\def\zQ{{\mathbb Q}} 
\def\zW{{\mathbb W}} 
\def\bK{{\bf K}}
\def\C{{\bf C}} 
\def\M{{\bf M}} 
\def\E{{\bf E}} 
\def\N{{\bf N}}
\def\O{{\bf O}} 
\def\bN{{\bf N}} 
\def\bX{{\bf X}} 
\def\GF{{\rm GF}} 
\def\cC{{\mathcal C}} 
\def\cA{{\mathcal A}}
\def\cB{{\mathcal B}} 
\def\cD{{\mathcal D}} 
\def\cE{{\mathcal E}} 
\def\cF{{\mathcal F}} 
\def\cG{{\mathcal G}} 
\def\cH{{\mathcal H}}
\def\cI{{\mathcal I}} 
\def\cL{{\mathcal L}} 
\def\cP{{\mathcal P}} 
\def\cQ{{\mathcal Q}} 
\def\cR{{\mathcal R}} 
\def\cRY{{\mathcal RY}}
\def\cS{{\mathcal S}} 
\def\cT{{\mathcal T}} 
\def\cU{{\mathcal U}} 
\def\cV{{\mathcal V}} 
\def\cW{{\mathcal W}} 
\def\cZ{{\mathcal Z}} 
\def\oB{{\ov B}}
\def\oC{{\ov C}} 
\def\ooB{{\ov{\ov B}}} 
\def\ozB{{\ov{\zB}}}
\def\ozD{{\ov{\zD}}} 
\def\ozG{{\ov{\zG}}}
\def\tcA{{\widetilde\cA}} 
\def\tcC{{\widetilde\cC}}
\def\tcF{{\widetilde\cF}} 
\def\tcI{{\widetilde\cI}}
\def\tB{{\widetilde B}} 
\def\tC{{\widetilde C}}
\def\tD{{\widetilde D}} 
\def\ttB{{\widetilde{\widetilde B}}}
\def\ttC{{\widetilde{\widetilde C}}}
\def\tba{{\tilde\ba}} 
\def\ttba{{\tilde{\tilde\ba}}}
\def\tbb{{\tilde\bb}} 
\def\ttbb{{\tilde{\tilde\bb}}}
\def\tbc{{\tilde\bc}} 
\def\tbd{{\tilde\bd}}
\def\tbe{{\tilde\be}} 
\def\tbt{{\tilde\bt}}
\def\tbu{{\tilde\bu}} 
\def\tbv{{\tilde\bv}}
\def\tbw{{\tilde\bw}} 
\def\tdl{{\tilde\dl}} 
\def\ocP{{\ov\cP}}
\def\tzA{{\widetilde\zA}} 
\def\tzC{{\widetilde\zC}}
\def\new{{\mbox{\footnotesize new}}}
\def\old{{\mbox{\footnotesize old}}}
\def\prev{{\mbox{\footnotesize prev}}}
\def\oo{{\mbox{\sf\footnotesize o}}}
\def\pp{{\mbox{\sf\footnotesize p}}}
\def\nn{{\mbox{\sf\footnotesize n}}} 
\def\oR{{\ov R}}
  
  
\def\gA{{\mathfrak A}} 
\def\gV{{\mathfrak V}} 
\def\gS{{\mathfrak S}} 
\def\gK{{\mathfrak K}} 
\def\gH{{\mathfrak H}}
  
\def\ba{{\bf a}} 
\def\bb{{\bf b}} 
\def\bc{{\bf c}} 
\def\bd{{\bf d}} 
\def\be{{\bf e}} 
\def\bbf{{\bf f}} 
\def\bg{{\bf g}}
\def\bh{{\bf h}}
\def\bi{{\bf i}} 
\def\bm{{\bf m}} 
\def\bo{{\bf o}} 
\def\bp{{\bf p}} 
\def\bs{{\bf s}} 
\def\bu{{\bf u}} 
\def\bt{{\bf t}} 
\def\bv{{\bf v}} 
\def\bx{{\bf x}}
\def\by{{\bf y}} 
\def\bw{{\bf w}} 
\def\bz{{\bf z}}
\def\ga{{\mathfrak a}} 
\def\oal{{\ov\al}} 
\def\obeta{{\ov\beta}}
\def\ogm{{\ov\gm}} 
\def\oep{{\ov\varepsilon}}
\def\oeta{{\ov\eta}} 
\def\oth{{\ov\th}} 
\def\ovm{{\ov\mu}}
\def\ozero{{\ov0}}
\def\bB{{\bf B}} 
\def\bA{{\bf A}}

  
\def\CCSP{\hbox{\rm c-CSP}} 
\def\CSP{{\rm CSP}} 
\def\NCSP{{\rm \#CSP}} 
\def\mCSP{{\rm MCSP}} 
\def\FP{{\rm FP}} 
\def\PTIME{{\bf PTIME}} 
\def\GS{\hbox{($*$)}} 
\def\ry{\hbox{\rm r+y}}
\def\rb{\hbox{\rm r+b}} 
\def\Gr#1{{\mathrm{Gr}(#1)}}
\def\Grp#1{{\mathrm{Gr'}(#1)}} 
\def\Grpr#1{{\mathrm{Gr''}(#1)}}
\def\Scc#1{{\mathrm{Scc}(#1)}} 
\def\rel{R} 
\def\relo{Q}
\def\rela{S} 
\def\reli{T} 
\def\relp{P} 
\def\dep{\mathsf{dep}}
\def\Filt{\mathrm{Ft}}
\def\Filts{\mathrm{Fts}} 
\def\Agr{$\mathbb{A}$}
\def\Al{\mathrm{Alg}}
\def\Sig{\mathrm{Sig}}
\def\strat{\mathsf{strat}}
\def\relmax{\mathsf{relmax}}
\def\srelmax{\mathsf{srelmax}}
\def\Meet{\mathsf{Meet}}
\def\amax{\mathsf{amax}}
\def\umax{\mathsf{umax}}
\def\emin{\mathsf{Z}}
\def\as{\mathsf{as}}
\def\star{\hbox{$(*)$}}
\def\bmal{{\mathbf m}}
\def\Af{\mathsf{Af}}
\let\sqq=\sqsubseteq
\def\maj{\mathsf{maj}}
\def\razm{\mathsf{size}}
\def\Razm{\mathsf{MAX}}
\def\Centr{\mathsf{Center}}
\def\centr{\mathsf{center}}

  
\let\sse=\subseteq 
\def\ang#1{\langle #1 \rangle}
\def\angg#1{\left\langle #1 \right\rangle}
\def\dang#1{\ang{\ang{#1}}} 
\def\vc#1#2{#1 _1\zd #1 _{#2}}
\def\tms{\tm\dots\tm}
\def\zd{,\ldots,} 
\let\bks=\backslash 
\def\red#1{\vrule height7pt depth3pt width.4pt
\lower3pt\hbox{$\scriptstyle #1$}}
\def\fac#1{/\lower2pt\hbox{$\scriptstyle #1$}}
\def\me{\stackrel{\mu}{\eq}} 
\def\nme{\stackrel{\mu}{\not\eq}}
\def\eqc#1{\stackrel{#1}{\eq}} 
\def\cl#1#2{\arraycolsep0pt
\left(\begin{array}{c} #1\\ #2 \end{array}\right)}
\def\cll#1#2#3{\arraycolsep0pt \left(\begin{array}{c} #1\\ #2\\
#3 \end{array}\right)} 
\def\clll#1#2#3#4{\arraycolsep0pt
\left(\begin{array}{c} #1\\ #2\\ #3\\ #4 \end{array}\right)}
\def\cllll#1#2#3#4#5#6{ \left(\begin{array}{c} #1\\ #2\\ #3\\
#4\\ #5\\ #6 \end{array}\right)} 
\def\pr{{\rm pr}}
\let\upr=\uparrow 
\def\ua#1{\hskip-1.7mm\uparrow^{#1}}
\def\sua#1{\hskip-0.2mm\scriptsize\uparrow^{#1}} 
\def\lcm{{\rm lcm}} 
\def\perm#1#2#3{\left(\begin{array}{ccc} 1&2&3\\ #1&#2&#3
\end{array}\right)} 
\def\w{$\wedge$} 
\let\ex=\exists
\def\NS{{\sc (No-G-Set)}} 
\def\lev{{\sf lev}}
\let\rle=\sqsubseteq 
\def\ryle{\le_{ry}} 
\def\ryprec{\le_{ry}}
\def\os{\mbox{[}} 
\def\zs{\mbox{]}}
\def\link{{\sf link}}
\def\solv{\stackrel{s}{\sim}} 
\def\mal{\mathbf{m}}
\def\precs{\prec_{as}}

  
\def\lb{$\linebreak$}  
  
\def\ar{\hbox{ar}} 
\def\Im{{\sf Im}\;} 
\def\deg{{\sf deg}}
\def\id{{\rm id}}
  
\let\al=\alpha 
\let\gm=\gamma 
\let\dl=\delta 
\let\ve=\varepsilon
\let\ld=\lambda 
\let\om=\omega 
\let\vf=\varphi 
\let\vr=\varrho
\let\th=\theta 
\let\sg=\sigma 
\let\Gm=\Gamma 
\let\Dl=\Delta
\let\kp=\kappa
  
  
\font\tengoth=eufm10 scaled 1200 
\font\sixgoth=eufm6
\def\goth{\fam12} 
\textfont12=\tengoth 
\scriptfont12=\sixgoth
\scriptscriptfont12=\sixgoth 
\font\tenbur=msbm10
\font\eightbur=msbm8 
\def\bur{\fam13} 
\textfont11=\tenbur
\scriptfont11=\eightbur 
\scriptscriptfont11=\eightbur
\font\twelvebur=msbm10 scaled 1200 
\textfont13=\twelvebur
\scriptfont13=\tenbur 
\scriptscriptfont13=\eightbur
\mathchardef\nat="0B4E 
\mathchardef\eps="0D3F

\title{A dichotomy theorem for nonuniform CSPs simplified}
\author{Andrei A.\ Bulatov\\ 
} 
\date{} 
\maketitle

\begin{abstract}
In a non-uniform Constraint Satisfaction problem $\CSP(\Gm)$, where $\Gm$
is a set of relations on a finite set $A$, the goal is
to find an assignment of values to variables subject to constraints imposed
on specified sets of variables using the relations from $\Gm$.
The Dichotomy Conjecture for the non-uniform CSP states that for every 
constraint language $\Gm$ the problem  $\CSP(\Gm)$ is either solvable in 
polynomial time or is NP-complete. It was proposed by Feder and Vardi 
in their seminal 1993 paper. In this paper we confirm the Dichotomy 
Conjecture.
\end{abstract}

\section{Introduction}

In a Constraint Satisfaction Problem (CSP) the question is to decide whether 
or not it is possible to satisfy a given set of 
constraints. One of the standard ways to specify a constraint is to require that 
a combination of values of a certain set of variables belongs to a given relation. 
If the constraints allowed in a problem have to come from some set $\Gm$ of 
relations, such a restricted problem is referred 
to as a \emph{nonuniform CSP} and denoted $\CSP(\Gm)$. The set $\Gm$ is 
then called a \emph{constraint language}. Nonuniform CSPs not only provide 
a powerful
framework ubiquitous across a wide range of disciplines from theoretical 
computer science to computer vision, but also admit natural and elegant 
reformulations such as the homomorphism problem, and characterizations, 
in particular, as the class of problems equivalent to a logic class MMSNP.
Many different versions of the CSP have been studied across various fields. 
These include CSPs over infinite sets, counting CSPs (and related Holant 
problem and the problem of computing partition functions), several variants
of optimization CSPs, valued CSPs, quantified CSPs, and numerous related 
problems. The 
reader is referred to the recent book \cite{Krokhin17:constraint} for a survey 
of the state-of-the art in some of these areas. In this paper we, however, focus 
on the decision nonuniform CSP and its complexity.

A systematic study 
of the complexity of nonuniform CSPs was started by Schaefer in 1978
\cite{Schaefer78:complexity} who showed that for every constraint language
$\Gm$ over a 2-element set the problem $\CSP(\Gm)$ is either solvable in 
polynomial time or is NP-complete. Schaefer also asked about the complexity of 
$\CSP(\Gm)$ for languages over larger sets. The next step in the study of 
nonuniform CSPs was made in the 
seminal paper by Feder and Vardi \cite{Feder93:monotone,Feder98:monotone}, 
who apart from considering numerous aspects of the problem, posed the 
\emph{Dichotomy Conjecture} that states that for every finite constraint language 
$\Gm$ over a finite set the problem $\CSP(\Gm)$ is either solvable in polynomial 
time or is NP-complete. This conjecture has become a focal point of the CSP 
research and most of the effort in this area revolves to some extent around the 
Dichotomy Conjecture.

The complexity of the CSP in general and the Dichotomy Conjecture in particular
has been studied by several research communities using a variety of methods, 
each contributing an important aspect of the problem.
The CSP has been an established area in artificial intelligence for decades, and 
apart from developing efficient general methods of solving CSPs researchers 
tried to identify tractable fragments of the problem \cite{Dechter03:processing}.
A very important special case of the CSP, the (Di)Graph Homomorphism 
problem and the $H$-Coloring problem have been actively studied in the 
graph theory community, see, e.g.\ \cite{Hell90:h-coloring,Hell04:homomorphism}
and subsequent works by Hell, Feder, Bang-Jensen, Rafiey and others. 
Homomorphism duality introduced in these works has been very useful in
understanding the structure of constraint problems. The CSP
plays a major role and has been successfully studied in database theory, logic and 
model theory \cite{Kolaitis03:csp,Kolaitis00:game,Gottlob14:treewidth}, although 
the version of the problem mostly used there is not necessarily nonuniform. 
Logic games and strategies are now a standard tool in most of CSP algorithms. An 
interesting approach to the Dichotomy Conjecture through long codes was
suggested by Kun and Szegedy \cite{Kun16:new}. Brown-Cohen and Raghavendra 
proposed to study the conjecture using techniques based on decay of 
correlations \cite{Brown-Cohen16:correlation}. In this paper we use the 
algebraic structure of the CSP, which is briefly discussed next.

The most effective approach to the study of the CSP turned out to be the 
\emph{algebraic approach} that associates 
every constraint language with its (universal) algebra of polymorphisms. This 
approach was first developed in a series of papers by Jeavons and coauthors 
\cite{Jeavons97:closure,Jeavons98:algebraic,Jeavons98:consist} and then refined
by Bulatov, Krokhin, Barto, Kozik, Maroti, Zhuk and others 
\cite{Barto12:absorbing,Barto12:near,Barto14:local,%
Bulatov05:classifying,Bulatov04:graph,Bulatov08:recent,Maroti11:Malcev,%
Maroti10:tree,Zhuk14:key,Zhuk16:7-element}. While the 
complexity of $\CSP(\Gm)$ has been already solved for some interesting 
classes of structures such as graphs \cite{Hell90:h-coloring}, the algebraic approach
allowed the researchers to confirm the Dichotomy Conjecture in a number of more 
general  cases: for languages over a set of size up to 7 
\cite{Bulatov02:3-element,%
Bulatov06:3-element,Markovic11:4-element,Zhuk16:7-element}, 
so called conservative languages 
\cite{Bulatov03:conservative,Bulatov11:conservative,%
Bulatov16:conservative,Barto11:conservative}, and some classes of digraphs
\cite{Barto09:sources}. It also helped to design the main classes of CSP 
algorithms \cite{Barto14:local,Bulatov06:simple,Bulatov16:restricted,%
Berman10:varieties,Idziak10:few}, and to refine the exact complexity of the CSP 
\cite{Allender05:refining,Barto12:near,Dalmau08:majority,Larose07:first-order}.

In this paper we confirm the Dichotomy Conjecture for arbitrary languages 
over finite sets. More precisely we prove the following

\begin{theorem}\label{the:main}
For any finite constraint language $\Gm$ over a finite set the problem 
$\CSP(\Gm)$ is either solvable in polynomial time or is NP-complete.
\end{theorem}

The same result has been independently obtained by Zhuk 
\cite{Zhuk17:proof,Zhuk17:proof-arxiv,Zhuk18:modification}. 

The proved criterion matches the algebraic form of the Dichotomy Conjecture 
suggested in \cite{Bulatov05:classifying}. The hardness part of the conjecture 
has been known for long time. Therefore the main achievement of this paper 
is a polynomial time algorithm for problems satisfying the tractability 
condition from \cite{Bulatov05:classifying}. 

Using the algebraic language we can state the result in a stronger form.
Let $\zA$ be a finite idempotent algebra and let $\CSP(\zA)$ denote the
union of problems $\CSP(\Gm)$ such that every term operation of $\zA$
is a polymorphism of $\Gm$. Problem $\CSP(\zA)$ is no longer a nonuniform
CSP, and Theorem~\ref{the:main} allows for problems $\CSP(\Gm)\sse\CSP(\zA)$
to have different solution algorithms even when $\zA$ meets the tractability 
condition. We show that the solution algorithm only depends on the algebra
$\zA$.

\begin{theorem}\label{the:main2}
For a finite idempotent algebra that satisfies the conditions of the 
Dichotomy Conjecture there is a uniform solution algorithm for $\CSP(\zA)$.
\end{theorem}

An interesting question arising from Theorems~\ref{the:main},\ref{the:main2}
is known as the \emph{Meta-problem}: Given a constraint language or a 
finite algebra, decide whether or not it satisfies the conditions of the theorems.
The answer to this question is not quite simple, for a thorough study of
the Meta-problem see \cite{Chen17:asking,Freese09:complexity}.

We start with introducing the terminology and notation for CSPs that is used 
throughout the paper and reminding the basics of the algebraic approach. 
Then in Section~\ref{sec:centralizer} we introduce the key ingredients 
used in the algorithm: separation of congruences and centralizers. Then in 
Section~\ref{sec:algorithm1} we apply these concepts to CSPs, first, 
to demonstrate how centralizers help to decompose an instance into smaller
subinstances, and, second, to introduce a new kind of minimality condition for 
CSPs, \emph{block minimality}. After that we state the main results used by
the algorithm and describe the algorithm itself. The last part of the paper, 
Sections~\ref{sec:compressed}--\ref{sec:non-affine}, is devoted to 
proving the technical results.

\section{CSP, universal algebra and the Dichotomy 
conjecture}\label{sec:csp-p1}

For a detailed introduction to the CSP and the algebraic approach to its structure
the reader is referred to a recent survey by Barto et al.\
\cite{Barto17:polymorphisms}. Basics of universal algebra can be learned from
the textbook \cite{Burris81:universal}. In preliminaries to this paper we 
therefore focus on what is needed for our result.

\subsection{The CSP}
The `AI' formulation of the CSP best fits our purpose.
Fix a finite set $A$ and let $\Gm$ be a \emph{constraint language} over $A$, 
that is, a set --- not necessarily finite --- of relations over $A$. The 
(\emph{nonuniform}) \emph{Constraint Satisfaction 
Problem} (\emph{CSP}) associated with language $\Gm$ 
is the problem $\CSP(\Gm)$, in which, an \emph{instance}  
is a pair $(V,\cC)$, where $V$ is a set of variables; and $\cC$ is a set of 
\emph{constraints}, 
i.e.\ pairs $\ang{\bs,\rel}$, where $\bs=(\vc vk)$ is a tuple of 
variables from $V$, the \emph{constraint scope}, 
and $\rel\in\Gm$, the $k$-ary 
\emph{constraint relation}. We always assume that 
relations are given explicitly by a list of tuples. The way constraints are 
represented does not matter if $\Gm$ is finite, but it may change the
complexity of the problems for infinite languages.
The goal is to find a \emph{solution}, i.e., a mapping 
$\vf:V\to A$ such that for every constraint $\ang{\bs,\rel}\in\cC$,
$\vf(\bs)\in\rel$.

\subsection{Algebraic methods in the CSP}

Jeavons et al.\ in \cite{Jeavons97:closure,Jeavons98:algebraic} were the first
to observe that higher order symmetries of constraint languages, called 
polymorphisms, play a significant role in  the study of the complexity of the CSP.
A \emph{polymorphism} of a relation $\rel$ over $A$ is an operation 
$f(\vc xk)$ on $A$ such that for any choice of $\vc\ba k\in\rel$ we have
$f(\vc\ba k)\in\rel$. If this is the case we also say that $f$ 
\emph{preserves} $\rel$, or that $\rel$
is \emph{invariant} with respect to $f$. A polymorphism 
of a constraint language $\Gm$ is an 
operation that is a polymorphism of every $\rel\in\Gm$. 

\begin{theorem}[\cite{Jeavons97:closure,Jeavons98:algebraic}]%
\label{the:algebra-csp}
For constraint languages $\Gm,\Dl$, where $\Gm$ is finite, if every 
polymorphism of $\Dl$ is also a polymorphism of $\Gm$, then 
$\CSP(\Gm)$ is polynomial time reducible to $\CSP(\Dl)$.\footnote{Using 
the $s-t$-Connectivity algorithm by Reingold \cite{Reingold08:undirect} this
reduction can be improved to a log-space one.}
\end{theorem}

Listed below are several types of polymorphisms that occur frequently 
throughout the paper. The presence of each of these polymorphisms
imposes strong restrictions on the structure of invariant relations that can be used 
in designing a solution algorithm. Some of such results will be mentioned 
later.\\[1mm]
-- \emph{Semilattice} operation is a binary operation $f(x,y)$ such that $f(x,x)=x$,
$f(x,y)=f(y,x)$, and $f(x,f(y,z))=f(f(x,y),z)$ for all $x,y,z\in A$;\\[1mm]
-- $k$-ary \emph{near-unanimity} operation is a $k$-ary operation 
$u(\vc xk)$
such that\lb $u(y,x\zd x)=u(x,y,x\zd x)=\dots =u(x\zd x,y)=x$ for all 
$x,y\in A$; a ternary near-unanimity operation $m$ is called a \emph{majority} 
operation, it satisfies the equations $m(y,x,x)=m(x,y,x)=m(x,x,y)=x$;\\[1mm]
-- \emph{Mal'tsev} operation is a ternary operation $h(x,y,z)$ satisfying the 
equations
$h(x,y,y)=h(y,y,x)=x$ for all $x,y\in A$; the \emph{affine} operation $x-y+z$ of
an Abelian group is a special case of a Mal'tsev operation;\\[1mm]
-- $k$-ary \emph{weak near-unanimity} operation is a $k$-ary operation $w$ that 
satisfies the same equations as a near-unanimity operation
$w(y,x\zd x)=\dots =w(x\zd x,y)$, except for the last one ($=x$).

\smallskip

To illustrate the effect of polymorphisms on the structure of invariant relations
we give a few examples that involve polymorphisms introduced 
above. First, we need some terminology and notation.

By $[n]$ we denote the set $\{1\zd n\}$. For sets $\vc An$ tuples 
from $A_1\tms A_n$ are denoted in boldface, say, $\ba$; the $i$th component of 
$\ba$ is referred to as $\ba[i]$. 
An $n$-ary relation $\rel$ over sets $\vc An$ is any subset of 
$A_1\tms A_n$. 
For $I=\{\vc ik\}\sse[n]$ by $\pr_I\ba,\pr_I\rel$ we denote the 
\emph{projections}\index{projections} $\pr_I\ba=(\ba[i_1]\zd\ba[i_k])$, 
$\pr_I\rel=\{\pr_I\ba\mid\ba\in\rel\}$ of tuple
$\ba$ and relation $\rel$. If $\pr_i\rel=A_i$ for each $i\in[n]$, relation $\rel$ is 
said to be a \emph{subdirect product} of 
$A_1\tms A_n$. Sometimes it is convenient to label the coordinate positions 
of relations by elements of some set other than $[n]$, e.g.\ by variables of a CSP.

\begin{example}\label{exa:polymorphisms}
(1) Let $\join$ be the binary operation of disjunction on $\{0,1\}$, as is easily 
seen, it is a semilattice operation. The following property of relations invariant 
under $\join$ helps solving the corresponding CSP: A relation $\rel$ contains 
the tuple $(1\zd1)$ whenever for each coordinate position $\rel$ contains a
tuple with a 1 in that position. Similarly, relations invariant under other 
semilattice operations on larger sets always contain a sort of a `maximal' 
tuple.\\[2mm]
(2) By the results of \cite{Baker75:chinese-remainder} a tuple $\ba$ belongs to
a ($n$-ary) relation $\rel$ invariant under a $k$-ary near-unanimity operation
if and only if for every $(k-1)$-element set $I\sse[n]$ we have $\pr_I\ba\in\pr_I\rel$.
In particular, if $f$ is the majority operation on $\{0,1\}$ given by 
$(x\meet y)\join(y\meet z)\join(z\meet x)$, and $\rel$ is a relation on $\{0,1\}$,
then $\ba\in\rel$ if and only if $(\ba[i],\ba[j])\in\pr_{ij}\rel$. This property 
easily gives rise to a reduction of the corresponding CSP to 2-SAT.\\[2mm]
(3) If $m(x,y,z)=x-y+z$ is the affine operation of, say, $\zZ_p$, $p$ prime, 
then relations invariant with respect to $m$ are exactly those that can be
represented as solution sets of systems of linear equations over $\zZ_p$,
and the corresponding CSP can be solved by Gaussian Elimination.
One direction is easy to see. If $\rel=\{\bx\mid \bx\cdot M=\bd\}$, where
$M$ is the matrix of the system of equations, and $\ba,\bb,\bc\in\rel$, then
$$
(\ba-\bb+\bc)\cdot M=\ba\cdot M-\bb\cdot M+\bc\cdot M=\bd-\bd+\bd=\bd,
$$
implying $m(\ba,\bb,\bc)\in\rel$. The other direction is more involved.
\EEX
\end{example}

The next step in discovering more structure behind nonuniform CSPs 
was made in \cite{Bulatov05:classifying}, where universal algebras 
were brought into the picture.
A \emph{(universal) algebra} is a pair $\zA=(A,F)$ 
consisting of a set $A$, the \emph{universe} of $\zA$, and a 
set $F$ of operations on $A$. Operations from $F$ (called \emph{basic}) 
together with operations that 
can be obtained from them by means of  composition are called the 
\emph{term} operations of $\zA$.

Algebras allow for a more general definition of CSPs than the one used 
above. Let $\CSP(\zA)$ denote the class of nonuniform CSPs 
$\{\CSP(\Gm)\mid \Gm\sse\Inv(F), \text{ $\Gm$ finite}\}$, where $\Inv(F)$
denotes the set of all relations invariant with respect to all operations from $F$. 
Note that 
the tractability of $\CSP(\zA)$ can be understood in two ways: as the existence of 
a polynomial-time algorithm for every $\CSP(\Gm)$ from this class, or as the 
existence of a uniform polynomial-time algorithm for all such problems. One of the
implications of our results is that these two types of tractability are the same.
From the formal standpoint we will use the stronger one.

\subsection{Structural features of universal algebras}
We use some structural elements of algebras, the main of which are 
subalgebras, congruences, and quotient algebras.
For $B\sse A$ and an operation $f$ on $A$ by $f\red B$ we denote the 
restriction of $f$ on $B$. Algebra $\zB=(B,\{f\red B\mid f\in F\})$ is 
a \emph{subalgebra} of $\zA$ if $f(\vc bk)\in B$ for any $\vc bk\in B$ and
any $f\in F$.

Congruences play a very significant role in our algorithm, and we discuss them 
in more detail. A \emph{congruence} is an equivalence relation $\al\in\Inv(F)$. 
This means that
for any operation $f\in F$ and any $(a_1,b_1)\zd (a_k,b_k)\in\al$ it 
holds\lb $(f(\vc ak),f(\vc bk))\in\al$. Hence one can define an 
algebra on $A\fac\al$, the set of $\al$-blocks, by setting 
$f\fac\al(a_1\fac\al\zd a_k\fac\al)=(f(\vc ak))\fac\al$ for $\vc ak\in A$, where
$a\fac\al$ denotes the $\al$-block containing $a$. The algebra $\zA\fac\al$ 
is called the \emph{quotient algebra modulo}~$\al$. Often the fact that $a,b$
are related by a congruence $\al$ is denoted $a\eqc\al b$.

\begin{example}\label{exa:congruences}
The following are examples of congruences and quotient algebras.\\[1mm]
(1) Let $\zA$ be any algebra. Then the equality relation $\zz_\zA$ and the full
binary relation $\zo_\zA$ on $\zA$ are congruences of $\zA$. The quotient
algebra $\zA\fac{\zz_\zA}$ is $\zA$ itself, while $\zA\fac{\zo_\zA}$ is a 
1-element algebra.\\[1mm]
(2) Let $\zL_n$ be an $n$-dimensional vector space and $\zL'$ its $k$-dimensional
subspace, $k\le n$. The binary relation $\pi$ given by: $(\ov a,\ov b)\in\pi$ iff  
$\ov a,\ov b$ have the same orthogonal projection on $\zL'$, is a
congruence of $\zL_n$ and $\zL_n\fac\pi$ is~$\zL'$.\\[1mm]
(3) The next example will be our running example throughout the paper. 
Let $A=\{0,1,2\}$, and let $\zA_M$
be the algebra with universe $A$ and two basic operations: a binary operation 
$r$ such that $r(0,0)=r(0,1)=r(2,0)=r(0,2)=r(2,1)=0$, 
$r(1,1)=r(1,0)=r(1,2)=1$,
$r(2,2)=2$; and a ternary operation $t$ such that $t(x,y,z)=x-y+z$ if 
$x,y,z\in\{0,1\}$, where $+,-$ are the operations of $\zZ_2$, $t(2,2,2)=2$,
and otherwise $t(x,y,z)=t(x',y',z')$, where $x'=x$ if $x\in\{0,1\}$ and 
$x'=0$ if $x=2$; the values $y',z'$ are obtained from $y,z$ by the same rule. 
It is an easy exercise to verify the following facts:
(a) $\zB=(\{0,1\},r\red{\{0,1\}},t\red{\{0,1\}})$ and 
$\zC=(\{0,2\},r\red{\{0,2\}},t\red{\{0,2\}})$ are subalgebras of $\zA_M$; 
(b) the partition $\{0,1\},\{2\}$ is a congruence of $\zA_M$, let us denote 
it $\th$; (c) algebra $\zC$ is basically a semilattice, that is, a set with a semilattice
operation, see Fig~\ref{fig:small-algebra}(a). 

The classes of congruence
$\th$ are $0\fac\th=\{0,1\}, 2\fac\th=\{2\}$. Then the quotient 
algebra
$\zA_M\fac\th$ is also basically a semilattice, as $r\fac\th(0\fac\th,0\fac\th)=
r\fac\th(0\fac\th,2\fac\th)=r\fac\th(2\fac\th,0\fac\th)=0\fac\th$ 
and~$r\fac\th(2\fac\th,2\fac\th)=2\fac\th$.
\EEX
\end{example}

\begin{figure}[ht]
\centerline{\includegraphics[totalheight=2.5cm,keepaspectratio]{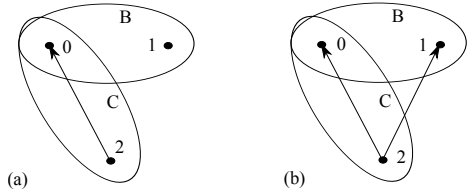}}
\caption{(a) Algebra $\zA_M$. (b) Algebra $\zA_N$. Dots represent elements, 
ovals represent subalgebras, and arrows represent semilattice edges (see 
Section~\ref{sec:few}).}\label{fig:small-algebra}
\end{figure}
\begin{figure}[ht]
\centerline{\includegraphics[totalheight=2.5cm,keepaspectratio]{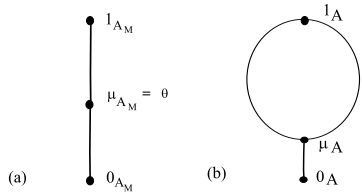}}
\caption{(a) The congruence lattice of algebra $\zA_M$; (b) congruence lattice of a 
subdirectly irreducible algebra.}\label{fig:congr}
\end{figure}

The (ordered) set of all congruences of $\zA$ is denoted by $\Con(\zA)$.
This set is actually a lattice, that is, the operations of meet $\meet$ and join
$\join$ can be defined so that $\al\meet\beta$ is the greatest lower bound of
$\al,\beta\in\Con(\zA)$ and $\al\join\beta$ is the least upper bound of $\al,\beta$. 
Fig.~\ref{fig:congr}(a) shows $\Con(\zA_M)$ for
the algebra $\zA_M$ from Example~\ref{exa:congruences}(3). By $\H\S(\zA)$ 
we denote the set of all quotient algebras of all subalgebras of $\zA$. 

\subsection{The Dichotomy Conjecture}

The results of \cite{Bulatov05:classifying} reduce the dichotomy conjecture  
to idempotent algebras. An algebra $\zA=(A,F)$ is said to be 
\emph{idempotent} if every operation $f\in F$
satisfies the equation $f(x\zd x)=x$. If $\zA$ is idempotent, then all the 
\emph{constant} relations $\{(a)\}$ are invariant under $F$. Therefore studying
CSPs over idempotent algebras is the same as studying the CSPs that allow all 
constant relations. Another useful property of idempotent algebras is that
every block of every its congruence is a subalgebra. We now can state
the algebraic version of the dichotomy theorem.

\begin{theorem}\label{the:algebra-dichot-p1}
For a finite idempotent algebra $\zA$ the following are equivalent:\\[.5mm]
(1) $\CSP(\zA)$ is solvable in polynomial time;\\[.5mm]
(2) $\zA$ has a weak near-unanimity term operation;\\[.5mm]
(3) every algebra from $\H\S(\zA)$ has a nontrivial term operation 
(that is, not a \emph{projection}, an operation of the form 
$f(\vc xk)=x_i$).\\[.5mm]
Otherwise $\CSP(\zA)$ is NP-complete.
\end{theorem}

The hardness part of this theorem is proved in \cite{Bulatov05:classifying}; the 
equivalence of (2) and (3) was proved in \cite{Bulatov01:varieties} and 
\cite{Maroti08:existence}. The equivalence of (1) to (2) and (3) is the main result 
of this paper. In the rest of the paper we assume all algebras to satisfy  
conditions (2),(3).

In fact, we will prove a slightly more general result. Let $\cA$ be a finite class 
of finite idempotent \emph{similar} algebras, that is, whose basic operations 
have the same `names' and the corresponding arities. One may assume that 
such a class is produced from a single algebra $\zA$ by taking subalgebras,
quotient algebras and also \emph{retractions} introduced in 
Section~\ref{sec:theorem-47}. Then $\CSP(\cA)$ denotes the class of
CSP instances whose variables can have different domains belonging to
$\cA$, see, e.g.\ \cite{Bulatov03:multi-sorted}. We will design an algorithm 
for $\CSP(\cA)$ whenever there is a near-unanimity term for all algebras 
in $\cA$ simultaneously.

\section{Bounded width and the few subpowers algorithm}%
\label{seg:bounded-width-p1}

Leaving aside occasional combinations thereof, there are only two standard 
types of algorithms solving the CSP. In this section we give a brief introduction 
into them.

\subsection{CSPs of bounded width}

Algorithms of the first kind are based on the idea of local propagation, that is
formally described below.

Let $\cP=(V,\cC)$ be a CSP instance. For $W\sse V$ by $\cP_W$ 
we denote the \emph{restriction} of 
$\cP$ onto $W$, that is, the instance 
$(W,\cC_W)$, where for each $C=\ang{\bs,\rel}\in\cC$, the set $\cC_W$
includes the constraint $C_W=\ang{\bs\cap W,\pr_{\bs\cap W}\rel}$,
where $\bs\cap W$ is the subtuple of $\bs$ containing all the elements 
from $W$ in $\bs$, say, $\bs\cap W=(\vc ik)$, and $\pr_{\bs\cap W}\rel$ 
stands for $\pr_{\{\vc ik\}}\rel$. 
The set of solutions of $\cP_W$ will be denoted by $\cS_W$. 

Unary solutions, that is, when $|W|=1$ play a special role. As is easily seen, 
for $v\in V$ the set $\cS_v$ is just the intersections of unary projections 
$\pr_v\rel$ of constraints whose scope contains $v$. Instance $\cP$ is said 
to be \emph{1-minimal} if for every $v\in V$ and every constraint
$C=\ang{\bs,\rel}\in\cC$ such that $v\in\bs$, it holds $\pr_v\rel=\cS_v$.
For a 1-minimal instance one may always 
assume that allowed values for a variable $v\in V$ is the set $\cS_v$. We call 
this set the \emph{domain} of $v$ and assume that CSP
instances may have different domains, which nevertheless are always 
subalgebras or quotient algebras of the original algebra $\zA$. It will be 
convenient to denote the domain of $v$ by $\zA_v$. The domain $\zA_v$
may change as a result of transformations of the instance.

Instance $\cP$ is said to be \emph{(2,3)-consistent} if it has a 
\emph{(2,3)-strategy}, that is, a 
collection of relations $\rel^X$, $X\sse V$, $|X|=2$ satisfying the following 
conditions (we use $\rel^v,\rel^{vw}$ for $\rel^{\{v\}},\rel^{\{v,w\}}$:\\
-- for every $X\sse V$ with $|X|\le2$, $\pr_{\bs\cap X}\rel^X\sse\cS_X$;\\ 
-- for every $X=\{u,v\}\sse V$, any $w\in V-X$, and any 
$(a,b)\in\rel^X$, there is $c\in\zA_w$ such that $(a,c)\in\rel^{uw}$
and $(b,c)\in\rel^{vw}$.\\[1mm]
\indent
Let the collection of relations $\rel^X$ be denoted by $\cR$. 
%
A tuple $\ba$ whose entries are indexed with elements 
of $W\sse V$ and such that $\pr_X\ba\in\rel^X$ for any $X\sse W$, $|X|=2$,
will be called \emph{$\cR$-compatible}. If a 
(2,3)-consistent instance $\cP$
with a (2,3)-strategy $\cR$ satisfies the additional condition\\
-- for every constraint $C=\ang{\bs,\rel}$ of $\cP$ every tuple $\ba\in\rel$
is $\cR$-compatible,\\[1mm]
it is called \emph{(2,3)-minimal}. For 
$k\in\nat$,  $(k,k+1)$-strategies, $(k,k+1)$-consistency, and 
$(k,k+1)$-minimality are defined in a similar way replacing 2,3 with $k,k+1$.

Instance $\cP$ is said to be \emph{minimal} (or 
\emph{globally minimal})  if for every 
$C=\ang{\bs,\rel}\in\cC$ and every $\ba\in\rel$ there is a solution 
$\vf\in\cS$ such that $\vf(\bs)=\ba$. Similarly, $\cP$ is said to be 
\emph{globally 1-minimal} if for every $v\in V$ and $a\in\zA_v$ there
is a solution $\vf$ with $\vf(v)=a$.

Any instance can be transformed to
a 1-minimal, (2,3)-consistent, or (2,3)-minimal instance in polynomial
time using the standard constraint propagation algorithms (see, e.g.\ 
\cite{Dechter03:processing}). These algorithms work by changing the constraint
relations and the domains of the variables eliminating some tuples and elements 
from them. We call such a process \emph{tightening} the 
instance. It is important to notice that if the original instance belongs to 
$\CSP(\zA)$ for some algebra $\zA$, that is, all its constraint relations are invariant
under the basic operations of $\zA$, the constraint relations obtained by 
propagation algorithms are also invariant under the basic operations of $\zA$, and 
so the resulting instance also belongs to $\CSP(\zA)$. Establishing minimality
amounts to solving the problem and therefore not always can be easily 
done.

If a constraint propagation algorithm solves a CSP, the problem is said to be of 
bounded width. More precisely, $\CSP(\Gm)$ (or $\CSP(\zA)$) is said to have 
\emph{bounded width} if for some $k$
every $(k,k+1)$-minimal instance from $\CSP(\Gm)$ (or $\CSP(\zA)$) has a 
solution. Problems 
of bounded width are very well studied, see an older survey 
\cite{Bulatov08:dualities} and a more recent paper \cite{Barto16:collapse}.

\begin{theorem}%
[\cite{Barto16:collapse,Bulatov16:restricted,Bulatov04:graph,%
Kozik16:characterization}]%
\label{the:bounded-width-p1}
For an idempotent algebra $\zA$ the following are equivalent:\\[1mm]
(1) $\CSP(\zA)$ has bounded width;\\[1mm]
(2) every (2,3)-minimal instance from $\CSP(\zA)$ has a solution;\\[1mm]
(3) $\zA$ has a weak near-unanimity term of arity
$k$ for every $k\ge3$;\\[1mm]
(4) every algebra $\H\S(\zA)$ has a nontrivial operation, and none of them 
is equivalent to a module (in a certain precise sense).
\end{theorem}

\subsection{Omitting semilattice edges and the few subpowers property}%
\label{sec:few}
%
The second type of CSP algorithms can be viewed as a generalization of 
Gaussian elimination, although, it utilizes just one property also used by 
Gaussian elimination: the set of solutions of a system of linear equations
or a CSP has a set of generators of size polynomial in the number of variables. 
The property 
that for every instance $\cP$ of $\CSP(\zA)$ its solution space $\cS$ has 
a set of generators of polynomial size is nontrivial, because there are only 
exponentially many such sets, while,
as is easily seen CSPs may have up to double exponentially 
many different sets of solutions. Formally, an algebra $\zA=(A,F)$ has 
\emph{few subpowers} if for every $n$ there are only exponentially many 
$n$-ary relations in $\Inv(F)$. 

Algebras with few subpowers are well studied and
the CSP over such an algebra has a polynomial-time solution algorithm, see, 
\cite{Berman10:varieties,Idziak10:few}. In particular, such algebras admit
a characterization in terms of the existence of a term operation with special
properties, an \emph{edge} term. We need only a subclass of
algebras with few subpowers that appeared in 
\cite{Bulatov16:restricted,Bulatov20:restricted} and is defined as follows.

A pair of elements $a,b\in\zA$ is said to be a \emph{semilattice edge}
if there is a binary term operation $f$ of $\zA$ such that $f(a,a)=a$ and
$f(a,b)=f(b,a)=f(b,b)=b$, that is, $f$ is a semilattice operation on $\{a,b\}$.
For example, the set $\{0,2\}$ from Example~\ref{exa:congruences}(3)
is a semilattice edge, and the operation $r$ of $\zA_M$ witnesses that.

\begin{prop}[\cite{Bulatov16:restricted,Bulatov20:restricted}]%
\label{pro:semilattice-edges-p1}
If an idempotent algebra $\zA$ has no semilattice edges, it has few subpowers,
and therefore $\CSP(\zA)$ is solvable in polynomial time.
\end{prop}

Semilattice edges have other useful properties including the following one 
that we use for reducing a CSP to smaller problems.

\begin{lemma}[Proposition~24, \cite{Bulatov20:graph}]%
\label{lem:multiplication-p1}
For any idempotent algebra $\zA$ there is a binary term operation $xy$ 
of $\zA$ (think multiplication) such that $xy$ is a semilattice operation 
on any semilattice edge and for any $a,b\in\zA$ either $ab=a$ or 
$\{a,ab\}$ is a semilattice edge.
\end{lemma}

Note that any semilattice operation satisfies the conditions of 
Lemma~\ref{lem:multiplication-p1}. The operation $r$ of the algebra
$\zA_M$ from Example~\ref{exa:congruences}(3) is not a semilattice 
operation (for instance, it does not satisfy the equation $r(x,y)=r(y,x)$),
but it satisfies the conditions of Lemma~\ref{lem:multiplication-p1}.

\section{Centralizers and decomposition of CSPs}\label{sec:centralizer}

In this section we introduce an alternative definition of the centralizer operator 
on congruence lattices studied in commutator theory, and study its 
properties and its connection to decompositions of CSPs.
Unlike the vast majority of the literature on the algebraic approach to the CSP
we use not only term operations, but also polynomial operations of an algebra.
It should be noted however that the first to use polynomials for CSP algorithms
was Maroti in \cite{Maroti10:tree}. We make use of some ideas from that paper
in the next section.

Let $f(\vc xk,\vc y\ell)$ be a $k+\ell$-ary term operation of an algebra 
$\zA=(A,F)$ and
$\vc b\ell\in\zA$. The operation $g(\vc xk)=f(\vc xk,\vc b\ell)$ is called a
\emph{polynomial} of $\zA$. The name `polynomial' refers to usual 
polynomials. Indeed, if $\zA$ is a ring, its polynomials as just defined 
are the same as polynomials in the regular sense. A polynomial that 
depends on only one variable, i.e.\ $k=1$, is said to be a \emph{unary} 
polynomial.

While polynomials of $\zA$ do not have to be polymorphisms of relations 
from $\Inv(F)$, congruences and unary polynomials are in a special 
relationship. More precisely, it is a well known fact that an equivalence 
relation over $\zA$ is a congruence if and only if it is preserved by all 
the unary polynomials of $\zA$. If $\al$ is a congruence, and $f$ is 
a unary polynomial, by $f(\al)$ we denote the set of pairs 
$\{(f(a),f(b))\mid (a,b)\in\al\}$.

\begin{example}\label{exa:more-small}
The unary polynomials of the algebra $\zA_M$ from 
Example~\ref{exa:congruences}(3) include the following unary operations
(these are the polynomials we will use, there are more unary polynomials of 
$\zA_M$):\\
$h_1(x)=r(x,0)=r(x,1)$, such that $h_1(0)=h_1(2)=0, h_1(1)=1$;\\
$h_2(x)=r(2,x)$, such that $h_2(0)=h_2(1)=0$, $h_2(2)=2$;\\
$h_3(x)=r(0,x)=0$.

The lattice $\Con(\zA_M)$ has 3 congruences: $\zz,\th,\zo$ (see 
Example~\ref{exa:congruences}(3)). As is easily seen,
$h_1(\th)\not\sse\zz$, $h_2(\zo)\not\sse\th$, but $h_1(\zo)\sse\th$, 
$h_2(\th)\sse\zz$, $h_3(\zo)\sse\zz$. 
\EEX
\end{example}

For an algebra $\zA$, a term operation $f(x,\vc yk)$, and $\ba\in\zA^k$, let
$f^\ba(x)=f(x,\ba)$. Let $\al,\beta\in\Con(\zA)$, $\al\le\beta$, 
and let $(\al:\beta)\sse\zA^2$ denote the greatest congruence such that  
for any term operation $f(x,\vc yk)$ and any $\ba,\bb\in\zA^k$
such that $(\ba[i],\bb[i])\in(\al:\beta)$, it holds that
$f^\ba(\beta)\sse\al$ if and only if $f^\bb(\beta)\sse\al$. Polynomials 
of the form $f^\ba,f^\bb$ are often called twin polynomials.

The congruence $(\al:\beta)$ will be called the \emph{centralizer} of 
$\al,\beta$\footnote{Traditionally, the centralizer of two congruences
is defined in a different way, see, e.g.\ \cite{Freese87:commutator}. 
Congruence $(\al:\beta)$ appeared in \cite{Hobby88:structure}, 
but completely inconsequentially, they did not study it at all, and its 
relation to the standard notion of centralizer remained unknown. 
We used the current definition in \cite{Bulatov17:dichotomy} and called 
it quasi-centralizer, again, not completely aware of its connection to 
the standard centralizer. 
Later Willard \cite{Willard19:centralizer} showed that the two concepts are 
equivalent, see \cite[Proposition~33]{Bulatov20:algebraic} for a proof, 
and we use `centralizer' here rather than `quasi-centralizer'.}. 
The following statement is one of the key ingredients of the algorithm.

\begin{lemma}[Corollary~37 \cite{Bulatov20:algebraic}]%
\label{lem:centralizer-multiplication}
Let $(\al:\beta)=\zo_\zA$, $a,b,c\in\zA$ and $b\eqc\beta c$.
Then $(ab,ac)\in\al$, where multiplication is as in 
Lemma~\ref{lem:multiplication-p1}.
\end{lemma}

\begin{example}\label{exa:more3-small}
In the algebra $\zA_M$, see Example~\ref{exa:congruences}(3), the 
centralizer
acts as follows: $(\zz:\th)=\zo$ and $(\th:\zo)=\th$. We start with the
second centralizer. Since every polynomial preserves congruences, for any term 
operation $h(x,\vc yk)$ and any $\ba,\bb\in\zA_M^k$ such that 
$(\ba[i],\bb[i])\in\th$ for $i\in[k]$, we have $(h^\ba(x),h^\bb(x))\in\th$ for 
any $x$. This of course implies $(\th:\zo)\ge\th$. On the other hand,
let $f(x,y)=r(y,x)$. Then 
\begin{eqnarray*}
& f^0(x)=f(x,0)=r(0,x)=h_3(x), &\\
& f^2(x)=f(x,2)=r(2,x)=h_2(x), &
\end{eqnarray*}
and $f^0(\zo)\sse\th$, while $f^2(\zo)\not\sse\th$. This means that 
$(0,2)\not\in(\th:\zo)$ and so $(\th:\zo)\subset\zo$.
For the first centralizer it suffices to demonstrate that the condition in
the definition of centralizer is satisfied for pairs of twin polynomials of
the form $(r(a,x),r(b,x))$, $(r(x,a),r(x,b))$, $(t(x,a_1,a_2),t(x,b_1,b_2))$,
$(t(a_1,x,a_2),t(b_1,x,b_2))$, $(t(a_1,a_2,x),\lb t(b_1,b_2,x))$ for 
$a,b,a_1,a_2,b_1,b_2\in\{0,1,2\}$, which can be 
verified directly.

Interestingly, Lemma~\ref{lem:centralizer-multiplication} implies that if we 
change the operation $r$ in just one point, it has a profound
effect on the centralizer $(\zz:\th)$. Let $\zA_N$ be the same 
algebra as $\zA_M$ with operations $r',t'$ defined in the same way as $r,t$, 
except $r'(2,1)=1$ replacing the value $r(2,1)=0$. In this case $\{1,2\}$
is also a semilattice edge, see 
Fig.~\ref{fig:small-algebra}(b). Let again $f(x,y)=r'(y,x)$ and $a=0,b=2$.
This time we have 
\begin{eqnarray*}
& f^0(x)=f(x,0)=r'(0,x)=h'_3(x), &\\
& f^2(x)=f(x,2)=r'(2,x)=h'_2(x), &
\end{eqnarray*}
where $h'_3(x)=0$ for all $x\in\{0,1,2\}$ and $h'_2(0)=0,h'_2(1)=1$
showing that $f^0(\th)\sse\zz$, while $f^2(\th)\not\sse\zz$.
\EEX
\end{example}

Fig.~\ref{fig:central-p1}(a),(b) shows the effect of large centralizers 
$(\al:\beta)$ on the structure of algebra $\zA$, which is a generalization 
of the phenomena observed in Example~\ref{exa:more3-small}. Dots there 
represent $\al$-blocks (assume $\al$ is the 
equality relation), ovals represent $\beta$-blocks, let they be $B$ and $C$,
and such that there is at least one semilattice edge between $B$ and $C$.
If $(\al:\beta)$ is the full relation, Lemmas~\ref{lem:multiplication-p1} 
and~\ref{lem:centralizer-multiplication} imply that for any $a\in B$ and
any $b,c\in C$ we have $ab=ac$, and so $ab$ is the only element of $C$
such that $\{a,ab\}$ is a semilattice edge (represented by arrows). In other 
words, we have a
mapping from $B$ to $C$ that can also be shown injective. We will use
this mapping to lift any solution with a value from $B$ to a solution with 
a value from $C$.

\begin{figure}[ht]
\centerline{\includegraphics[totalheight=3cm,keepaspectratio]{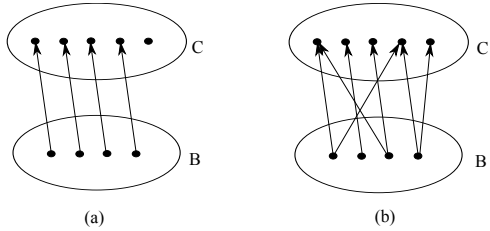}}
\caption{(a) $(\al:\beta)$ is the full relation; (b) $(\al:\beta)$ 
is not the full relation}\label{fig:central-p1}
\end{figure}

Finally, we prove an easy corollary from Lemma~\ref{lem:centralizer-multiplication}.

\begin{corollary}\label{cor:semilattice-in-centralizer}
Let $\al,\beta\in\Con(\zA)$, $\al\le\beta$, be such that $(\al:\beta)\ge\beta$.
Then for every $\beta$-block $B$ if $ab$ is a semilattice edge and $a,b\in B$, then
$a\eqc\al b$. 
\end{corollary}

\begin{proof}
Let $a,b\in\zB$, $a\not\eqc\al b$, form a semilattice edge, that is, $ab=ba=b$. However,
since $a\eqc{(\al:\beta)}b$, by Lemma~\ref{lem:centralizer-multiplication}
it must hold $aa\eqc\al bb$, a contradiction.
\end{proof}

\section{The algorithm}\label{sec:algorithm1}

In this section we introduce the reductions used in the algorithm, and then 
explain the algorithm itself. The reductions heavily use the algebraic structure 
of the domains of an instance, and the structure of the instance itself.

\subsection{Decomposition of CSPs}

We have seen in the previous section that large centralizers impose strong 
restrictions on the structure of an algebra. We start this section showing that
small centralizers imply certain properties of CSP instances, as well.

Let $\rel$ be a binary relation, a subdirect product of $\zA\tm\zB$, and 
$\al\in\Con(\zA)$, $\gm\in\Con(\zB)$. Relation $\rel$ is said to be 
\emph{$\al\gm$-aligned} 
if, for any $(a,c),(b,d)\in\rel$, $(a,b)\in\al$ if and only if $(c,d)\in\gm$. 
This means that if $\vc Ak$ are the $\al$-blocks of $\zA$, then there are
also $k$ $\gm$-blocks of $\zB$ and they can be labeled $\vc Bk$ in such 
a way that 
$$
\rel=(\rel\cap(A_1\tm B_1))\cup\dots\cup(\rel\cap(A_k\tm B_k)).
$$

This definition provides a way to decompose CSP instances.
Let $\cP=(V,\cC)$ be a (2,3)-minimal instance from $\CSP(\zA)$. We will always 
assume that a (2,3)-consistent or (2,3)-minimal instance has a constraint 
$C^X=\ang{X,\rel^X=\cS_X}$ for every $X\sse V$, $|X|\le2$. So, 
$\cC$ contains a constraint $C^{vw}=\ang{(v,w),\rel^{vw}}$ for 
every $v,w\in V$, and these relations form a (2,3)-strategy for $\cP$. 
Recall that $\zA_v$ denotes the domain of $v\in V$. Let $W\sse V$ and 
$\al_v\in\Con(\zA_v)$, $v\in W$, be such that for any $v,w\in W$ the
relation $\rel^{vw}$ is $\al_v\al_w$-aligned. The set $W$ is then
called a \emph{strand} of $\cP$. We will also say that $\cP_W$ is
\emph{$\ov\al$-aligned}.

For a strand $W$ and congruences $\al_v$ as above there is 
a one-to-one correspondence between $\al_v$- and $\al_w$-blocks of 
$\zA_v$ and $\zA_w$, $v,w\in W$. Moreover, by (2,3)-minimality these 
correspondences are consistent,
that is, if $u,v,w\in W$ and $B_u,B_v,B_w$ are $\al_u$-, $\al_v$-
and $\al_w$-blocks, respectively, such that 
$\rel^{uv}\cap(B_u\tm B_v)\ne\eps$ and 
$\rel^{vw}\cap(B_v\tm B_w)\ne\eps$, then 
$\rel^{uw}\cap(B_u\tm B_w)\ne\eps$. This means that $\cP_W$
can be split into several instances, whose domains are $\al_v$-blocks.

\begin{lemma}\label{lem:central-decomposition-p1}
Let $\cP,W,\al_v$ for each $v\in W$, be as above. Then $\cP_W$ can be 
decomposed into a collection of instances $\cP_1\zd\cP_k$, $k$ constant, 
$\cP_i=(W,\cC_i)$ such that 
every solution of $\cP_W$ is a solution of one of the $\cP_i$ and for every 
$v\in W$ its domain in $\cP_i$ is an $\al_v$-block.
\end{lemma}

\begin{example}\label{exa:more4-small}
Let $\zA_M$ be the algebra introduced in Example~\ref{exa:congruences}(3),
and $\rel$ is the following ternary relation over $\zA_M$ invariant under 
$r,t$, given by
$$
\rel=\left(\begin{array}{cccccccccc}
0&0&1&1&0&0&1&1&2&2\\ 0&1&1&0&0&1&1&0&2&2\\ 
0&0&0&0&1&1&1&1&0&2 
\end{array}\right),
$$
where triples, the elements of the relation are written vertically. 
Consider the following simple CSP instance from  $\CSP(\zA_M)$:
$\cP=(V=\{v_1,v_2,v_3,v_4,v_5\}, \{C^1=\ang{\bs_1=(v_1,v_2,v_3),\rel_1},
C^2=\ang{\bs_2=(v_2,v_4,v_5),\rel_2}\}$, where $\rel_1=\rel_2=\rel$. 
To make the instance (2,3)-minimal we 
run the appropriate local propagation algorithm on it. First, such an algorithm
adds new binary constraints 
$C^{v_iv_j}=\ang{(v_i,v_j),\rel^{v_iv_j}}$ for $i,j\in[5]$ starting
with $\rel^{v_iv_j}=\zA_M\tm\zA_M$. It then iteratively removes pairs from 
these relations that do not satisfy the (2,3)-minimality condition. 
Similarly, it tightens the original constraint relations if they violate the
conditions of (2,3)-minimality. It is not hard to see that
this algorithm does not change constraints $C^1,C^2$, and that the new 
binary relations are as follows: 
$\rel^{v_1v_2}=\rel^{v_2v_4}=\rel^{v_1v_4}=\th$,
$\rel^{v_1v_3}=\rel^{v_2v_3}=\rel^{v_2v_5}=\rel^{v_4v_5}=\relo$, 
and $\rel^{v_1v_5}=\rel^{v_3v_4}=\rel^{v_3v_5}=\rela$, where
\begin{eqnarray*}
\relo = \pr_{13}\rel &=& \left(\begin{array}{cccccc}
0&0&1&1&2&2\\ 0&1&0&1&0&2
\end{array}\right),\\
\rela &=& \left(\begin{array}{ccccccc}
0&0&1&1&0&2&2\\ 0&1&0&1&2&0&2
\end{array}\right).
\end{eqnarray*}
In order to distinguish 
elements and congruences of domains belonging to different variables let the 
domain of $v_i$ be denoted by $\zA_i$, its elements by $0_i,1_i,2_i$, and 
the congruences of $\zA_i$ by $\zz_i,\th_i,\zo_i$.

\begin{figure}[ht]
\centerline{\includegraphics[totalheight=3cm,keepaspectratio]{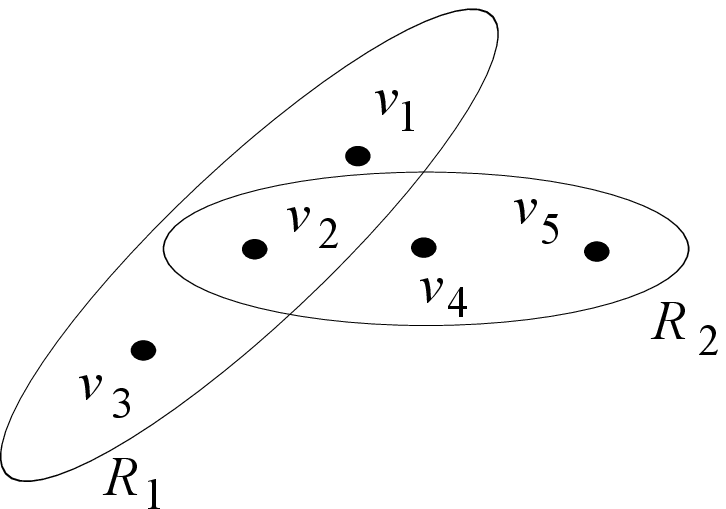}}
\caption{Instance $\cP$ from Example~\ref{exa:more4-small}}%
\label{fig:instance}
\end{figure}

Let $W=\{v_1,v_2,v_4\}$, 
$\al_i=\th_i$ for $v_i\in W$. Then, since $\rel^{v_1v_2}=\rel^{v_2v_4}
=\rel^{v_1v_4}=\th$ and therefore are $\al_i\al_j$-aligned, $i,j\in\{1,2,4\}$,
$W$ is a strand of $\cP$.  Therefore the instance
$\cP_W=(\{v_1,v_2,v_4\},\{C^1_W=\ang{(v_1,v_2),\pr_{v_1v_2}\rel_1},
C^2_W=\ang{(v_2,v_4),\pr_{v_2v_4}\rel_2}\})$ can be decomposed into
a disjoint union of two instances
\begin{eqnarray*}
\cP_1 &=& ( \{v_1,v_2,v_4\},\{\ang{(v_1,v_2),\relo_1},
\ang{(v_2,v_4),\relo_2}),\\
\cP_2 &=& ( \{v_1,v_2,v_4\},\{\ang{(v_1,v_2),\rela_1},
\ang{(v_2,v_4),\rela_2}),
\end{eqnarray*}
where $\relo_1=\{0_1,1_1\}\tm\{0_2,1_2\},
\relo_2=\{0_2,1_2\}\tm\{0_4,1_4\}$, $\rela_1=\{(2_1,2_2)\}, 
\rela_2=\{(2_2,2_4)\}$.
\EEX
\end{example}

\subsection{Irreducibility}
In order to formulate the algorithm properly we need one more transformation 
of algebras. An algebra $\zA$ is said to be 
\emph{subdirectly irreducible}
if the intersection of all its nontrivial (different from the equality relation)
congruences is nontrivial. This smallest nontrivial congruence $\mu_\zA$ is called 
the \emph{monolith} of $\zA$, see Fig.~\ref{fig:congr}(b). For instance,
the algebra $\zA_M$ from Example~\ref{exa:congruences}(3) is subdirectly
irreducible, because it has the smallest nontrivial congruence, $\th$.
It is a folklore observation that any CSP instance can be transformed in 
polynomial time to an instance, in which the domain of every variable 
is a subdirectly irreducible algebra. We will assume this property of all the
instances we consider.

\subsection{Block-minimality}

Using Lemma~\ref{lem:central-decomposition-p1} we introduce a new type of 
consistency of a CSP instance, block-minimality, which will be crucial 
for our algorithm. In a certain sense it is similar to the standard 
local consistency notions, as it also defined through a family of relations that have
to be consistent in a certain way. However, block-minimality is not quite 
local, and is more difficult to establish, as it involves solving smaller 
CSP instances recursively.
The definitions below are designed to allow for an efficient 
procedure to establish block-minimality. This is achieved either by allowing for 
decomposing a subinstance into instances over smaller domains as in 
Lemma~\ref{lem:central-decomposition-p1}, or by replacing large domains with 
their quotient algebras.

Let $\al_v$ be a congruence of $\zA_v$ for $v\in V$. By 
$\cP\fac{\ov\al}$ we denote
the instance $(V,\cC_{\ov\al})$ constructed as follows: the domain of 
$v\in V$ is $\zA_v\fac{\al_v}$; for every constraint $C=\ang{\bs,\rel}\in\cC$,
$\bs=(\vc vk)$, the set $\cC_{\ov\al}$ includes the constraint 
$\ang{\bs,\rel\fac{\ov\al}}$, where 
$\rel\fac{\ov\al}=\{(\ba[v_1]\fac{\al_{v_1}}\zd\ba[v_k]\fac{\al_{v_k}})
\mid \ba\in\rel\}$.

\begin{example}\label{exa:more5-small}
Consider the instance $\cP$ from Example~\ref{exa:more4-small}, and
let $\al_{v_i}=\th_i$ for each $i\in[5]$. Then $\cP\fac{\ov\al}$ is the instance
over $\zA_M\fac\th$ given by $\cP\fac{\ov\al}=
(V,\{\ang{\bs_1,\rel_1\fac{\ov\al}},\ang{\bs_2,\rel_2\fac{\ov\al}}\})$, where
$$
\rel_1\fac{\ov\al}=\rel_2\fac{\ov\al}=\left(\begin{array}{ccc}
0\fac\th & 2\fac\th & 2\fac\th \\ 0\fac\th & 2\fac\th & 2\fac\th \\ 
0\fac\th & 0\fac\th & 2\fac\th 
\end{array}\right).
$$
 \EEX
\end{example}

Let $\cP=(V,\cC)$ be a (2,3)-minimal instance, and 
for $X\sse V$, $|X|\le2$, there is a constraint $C^X=\ang{X,\rel^X}$, where 
$\rel^X$ is the set of partial solutions on $X$. 

Recall that an algebra $\zA_v$ is said to be semilattice free if it does not contain 
semilattice edges. Let $\razm(\cP)$ denote the 
maximal size of domains of $\cP$ that are not semilattice free and 
$\Razm(\cP)$ be the set of variables $v\in V$ such that
$|\zA_v|=\razm(\cP)$ and $\zA_v$ is not semilattice free. Finally, for $Y\sse V$ 
let $\mu^Y_v=\mu_v$ if $v\in Y$ and $\mu^Y_v=\zz_v$ otherwise. 

Instance $\cP$ is said to be \emph{block-minimal} if  
\begin{itemize}
\item[(BM)]
for every strand $U\sse V$ the problem $\cP_{/U}=\cP\fac{\ov\mu^Y}$,
where $Y=\Razm(\cP)-U$, is minimal.
\end{itemize}
The definition of block-minimality is designed in such a way that 
block-minimality can 
be efficiently established. Observe that a strand can be 
large, even equal to $V$. However $\cP_{/U}$ splits into a union of 
disjoint problems over smaller domains. 

\begin{example}\label{exa:more6-small}
Let us consider again the instance $\cP$ from Example~\ref{exa:more4-small}.
In that example we found all its binary solutions, and now we use 
them to find strands and to verify that this instance is block-minimal.
As we saw in Example~\ref{exa:more4-small}, unless $i,j\in\{1,2,4\}$ the
relation $\rel^{v_iv_j}$ is not $\al\beta$-aligned for any congruences 
$\al,\beta$ except the full ones. This means that the only strands of $\cP$
are $W=\{v_1,v_2,v_4\}$ and all the 1-element sets of variables.

Now we check the condition (BM) for $\cP$. Consider 
$W$. For this strand we have $Y=\{3,5\}$, and so 
$\mu^Y_1=\mu^Y_2=\mu^Y_4=\zz$ and $\mu^Y_3=\mu^Y_5=\th$. 
The problem $\cP_{/W}$ is the following problem: 
$(V,\{C'_1,C'_2\})$, where $C'_1=\ang{\bs_1,\rel^\th}$, 
$C'_2=\ang{\bs_2,\rel^\th}$, and
\[
R^\th = \left(\begin{array}{cccccc}
0&0&1&1&2&2\\ 0&1&0&1&2&2\\ 
0\fac\th&0\fac\th&0\fac\th&0\fac\th&0\fac\th&2\fac\th
\end{array}\right).
\]
Now, consider first $C_1$. For any tuple $(a_1,a_2,a_3)\in\rel^\th$, that is, 
assignment $v_1=a_1\in\zA_M, v_2=a_2\in\zA_M, v_3=a_3\in\zA_M\fac\th$,
we can extend this assignment to $v_4=v_2$ and $v_5=0\fac\th$ to obtain
a satisfying assignment of $\cP_{/W}$. For $C_2$ the argument is the same.

For 1-element strands consider $\{v_2\}$. Then $Y=\{v_1,v_3,v_4,v_5\}$, and 
$\mu^Y_1=\mu^Y_3=\mu^Y_4=\mu^Y_5=\th$. We have 
$\cP_{/\{v_2\}}=(V,\{C''_1,C''_2\})$, where $C''_1=\ang{\bs_1,\rel^{\th\th}_1}$, 
$C'_2=\ang{\bs_2,\rel^{\th\th}_2}$, and
\[
\rel^{\th\th}_1=\left(\begin{array}{cccc}
0\fac\th&0\fac\th&2\fac\th&2\fac\th\\  0&1&2&2\\ 
0\fac\th&0\fac\th&0\fac\th&2\fac\th
\end{array}\right), \qquad
\rel^{\th\th}_2=\left(\begin{array}{cccc}
0&1&2&2\\ 0\fac\th&0\fac\th&2\fac\th&2\fac\th\\ 
0\fac\th&0\fac\th&0\fac\th&2\fac\th
\end{array}\right).
\]
As is easily seen, any assignment to $v_1,v_2,v_3$ or to $v_2,v_4,v_5$ can be
extended to a solution of $\cP_{/\{v_2\}}$.
\EEX
\end{example}

For an instance $\cP$ we say that an instance $\cP'$ is \emph{strictly smaller} 
than instance $\cP$ if $\razm(\cP')<\razm(\cP)$. 

\begin{lemma}\label{lem:to-block-minimality}
Let $\cP=(V,\cC)$ be a (2,3)-minimal instance. 
Then $\cP$ can be transformed to an equivalent block-minimal 
instance $\cP'$ by solving a quadratic number of strictly smaller CSPs.
\end{lemma}

\begin{proof}
To establish block-minimality of $\cP$, for every strand 
$U\sse V$, we need to check if the problem given in condition (BM) is minimal.
If they are then $\cP$ is block-minimal, otherwise some tuples can be 
removed from some constraint relation $\rel$ (the set of tuples that remain 
in $\rel$ is always a subalgebra, as is easily seen), and the instance $\cP$ 
tightened, in which case we need to repeat the procedure with the tightened 
instance. Therefore we just need to show how to reduce solving those subproblems 
to solving strictly smaller CSPs. 

By the definition of a strand there is a partition $B_{w1}\zd B_{w\ell}$ 
of $\zA_w$ for $w\in U$ such that for every constraint
$\ang{\bs,\rel}\in\cC$, for any $w_1,w_2\in\bs\cap U$, any 
$\bb\in\rel$, and any $i\in[\ell]$ it holds $\bb[w_1]\in B_{w_1i}$ if and 
only if $\bb[w_2]\in B_{w_2i}$. Then the problem $\cP_{/U}$ is a disjoint 
union of instances $\vc\cP\ell$ given by: $\cP_i=(V,\cC_i)$, where
for every constraint $C=\ang{\bs,\rel}\in\cC$ there is 
$C_i=\ang{\bs,\rel_i}\in\cC_i$ such that 
\[
\rel_i=\{\ba'\mid \ba\in\rel, \ba[w]\in B_{wi} \text{ for each $w\in\bs\cap U$}\},
\]
with $\ba'[u]=\ba[u]\fac{\mu^Y_u}$, $Y=\Razm(\cP)-U$, for each $u\in\bs$.
Clearly, $\razm(\cP_i)<\razm(\cP)$ for each $i\in[\ell]$. 

In order to establish the minimality of $\cP_{/U}$ it suffices to do the
following. Take $C=\ang{\bs,\rel}\in\cC$ and $\ba\in\rel$. We need to check 
that $\ba'=\ba\fac{\ov\mu^Y}$, $Y=\Razm(\cP)-U$, extends to a solution 
of at least one of the problems $\vc\cP\ell$. For $i\in[\ell]$ let $\cP_i'$ 
be the problem obtained from $\cP_i$ as follows: fix the values of variables 
from $\bs$ to those of $\ba'$, or in other words, add the constraint 
$\ang{(w),\{\ba[w]\fac{\mu^Y_w}\}}$ for each $w\in\bs$. Then $\ba'$
can be extended to a solution of $\cP_i$ if and only if $\cP'_i$ has a 
solution.
\end{proof}

\subsection{The algorithm}\label{sec:the-algorithm}

We are now in a position to describe our solution algorithm. In the algorithm 
we distinguish three cases depending on the presence of 
semilattice edges and 
centralizers of the domains of variables. In each case we employ different 
methods of solving or reducing the instance to a strictly smaller one. 
Algorithm~\ref{alg:csp}, {\sf SolveCSP}, gives a more formal description of 
the solution algorithm.

Let $\cP=(V,\cC)$ be a subdirectly irreducible (2,3)-minimal instance. Let 
$\Centr(\cP)$ denote the set of variables $v\in V$ such 
that $(\zz_v:\mu_v)=\zo_v$. Let $\mu^*_v=\mu_v$ if 
$v\in\Razm(\cP)\cap\Centr(\cP)$ and $\mu^*_v=\zz_v$ otherwise. 

\paragraph{Semilattice free domains.}
If all domains of $\cP$ are semilattice free then $\cP$ can be solved in 
polynomial time, using the few subpowers algorithm, as shown in 
\cite{Idziak10:few,Bulatov16:restricted}.

\paragraph{Small centralizers}
If $\mu^*_v=\zz_v$ for all $v\in V$, by Theorem~\ref{the:non-central} 
block-minimality guarantees that a 
solution exists, and we can use Lemma~\ref{lem:to-block-minimality} to solve
the instance. 

\begin{theorem}\label{the:non-central}
If $\cP$ is subdirectly irreducible, (2,3)-minimal, block-minimal, and 
$\Razm(\cP)\cap\Centr(\cP)=\eps$, then $\cP$ has a solution.
\end{theorem}

\paragraph{Large centralizers}
Suppose that $\Razm(\cP)\cap\Centr(\cP)\ne\eps$. In this case the algorithm
proceeds in three steps.\\[1mm]
\emph{Stage 1.} 
Consider the problem $\cP\fac{\ov\mu^*}$. We establish the global 1-minimality 
of this
problem. If it is tightened in the process, we start solving the new problem
from scratch. To check global 1-minimality, for each $v\in V$ and every 
$a\in\zA_v\fac{\mu^*_v}$, we need to find a solution of the instance, or show it
does not exists. To this end, add the constraint $\ang{(v),\{a\}}$ to 
$\cP\fac{\ov\mu^*}$. The resulting problem belongs to $\CSP(\zA)$, since
$\zA_v$ is idempotent, and hence $\{a\}$ is a subalgebra of $\zA_v\fac{\mu^*_v}$.
Then we establish (2,3)-minimality and block minimality of the resulting problem.
Let us denote it $\cP'$. There are two possibilities. First, if 
$\razm(\cP')<\razm(\cP)$ then $\cP'$ is a problem strictly smaller than $\cP$
and can be solved by recursively calling Algorithm~\ref{alg:csp} on $\cP'$. If
$\razm(\cP')=\razm(\cP)$ then, as all the domains $\zA_v$ of maximal size 
for $v\in\Centr(\cP)$ are replaced with their quotient algebras, there is 
$w\not\in\Centr(\cP)$ such that $|\zA_w|=\razm(\cP)$ and $\zA_w$ is not semilattice
free. Therefore for every $u\in\Centr(\cP')$, for the corresponding domain 
$\zA'_u$ we have $|\zA'_u|<\razm(\cP)=\razm(\cP')$. Thus, 
$\Razm(\cP')\cap\Centr(\cP')=\eps$, and $\cP'$ has a solution by 
Theorem~\ref{the:non-central}.\\[1mm]
\emph{Stage 2.}
For every $v\in\Razm(\cP)$ we find a solution $\vf$ of $\cP\fac{\ov\mu^*}$ 
such that there is $a\in\zA_v$ such that 
$\{a,\vf(v)\}$ is a semilattice edge if $\mu^*_v=\zz_v$, or, if $\mu^*_v=\mu_v$,
there is $b\in\vf(v)$ such that $\{a,b\}$ is a semilattice edge. Take
$v\in\Razm(\cP)$ and $b\in\zA_v\fac{\mu^*_v}$ such that $\{a,b\}$ is a semilattice
edge in $\zA_v\fac{\mu^*_v}$ for some $a\in\zA_v\fac{\mu^*_v}$. Such a
semilattice edge exists, because $\zA_v$ is not semilattice free. Also, 
if $\mu^*_v\ne\zz_v$, then $v\in\Centr(\cP)$ and $(\zz_v:\mu_v)=\zo_v$ 
and by Corollary~\ref{cor:semilattice-in-centralizer} its semilattice edges are 
all between $\mu_v$-blocks. Since $\cP\fac{\ov\mu^*}$ is globally 1-minimal, 
there is a solution $\vf_{v,b}$ such that $\vf_{v,b}(v)=b$, and therefore 
$\vf_{v,b}$ satisfies the condition. Let $\Razm(\cP)=\{\vc v\ell\}$ and 
$\vc b\ell$ the values satisfying the requirements above.\\[1mm]
\emph{\sc Stage 3.}
We apply the transformation of $\cP$ suggested by Maroti in 
\cite{Maroti10:tree}. For a solution $\vf$ of $\cP\fac{\ov\mu^*}$, 
by $\cP\cdot\vf$
we denote the instance $(V,\cC_\vf)$ given by the rule: for every 
$C=\ang{\bs,\rel}\in\cC$ the set $\cC_\vf$ contains a constraint 
$\ang{\bs,\rel\cdot\vf}$. To construct $\rel\cdot\vf$ choose a tuple $\bb\in\rel$
such that $\bb[v]\fac{\mu^*_v}=\vf(v)$ for all $v\in\bs$; this is possible 
because $\vf$ is a solution of $\cP\fac{\ov\mu^*}$. Then set 
$\rel\cdot\vf=\{\ba\cdot\bb\mid \ba\in\rel\}$. By the results of 
\cite{Maroti10:tree} and Lemma~\ref{lem:centralizer-multiplication} 
the instance $\cP\cdot\vf$ has a solution if and only if $\cP$ does.
We now use the solutions $\vf_{v_1,b_1}\zd\vf_{v_\ell,b_\ell}$ to construct
a new problem 
\[
\cP^1=(\dots((\cP\cdot\vf_{v_1,b_1})\cdot\vf_{v_2,b_2})\cdot\dots)
\cdot\vf_{v_\ell,b_\ell}.
\]
Note that the transformation of $\cP$ above boils down to a collection of mappings
$p_v:\zA_v\to\zA_v$, $v\in V$, so called \emph{consistent mappings}, see 
Section~\ref{sec:theorem-47}, that also satisfy some additional properties. If we 
now repeat the procedure above starting from $\cP^1$ and using the same
solutions $\vf_{v_i,b_i}$, we obtain an instance $\cP^2$, for which the corresponding
collection of consistent mappings is $p_v\circ p_v$, $v\in V$. More generally,
\[
\cP^{i+1}=(\dots((\cP^i\cdot\vf_{v_1,b_1})\cdot\vf_{v_2,b_2})\cdot\dots)
\cdot\vf_{v_\ell,b_\ell}.
\]
There is $k$ such that $p^k_v$ is idempotent for every $v\in V$, that is,
$(p^k_v\circ p^k_v)(x)=p^k_v(x)$ for all $x$. Set $\cP^\dg=\cP^k$.
We will show later that $\razm(\cP^\dg)<\razm(\cP)$.

This last case can be summarized as the following

\begin{theorem}\label{the:central}
If $\cP\fac{\ov\mu^*}$ is globally 1-minimal, 
then $\cP$ can be reduced in polynomial time to a strictly smaller instance
over a class of algebras satisfying the conditions of the Dichotomy Conjecture.
\end{theorem}

\begin{algorithm}
\caption{Procedure {\sf SolveCSP}}
\label{alg:csp}
\begin{algorithmic}[1] 
\REQUIRE A CSP instance $\cP=(V,\cC)$ over $\cA$
\ENSURE A solution of $\cP$ if one exists, `NO' otherwise
\IF{all the domains are semilattice free}
\STATE Solve $\cP$ using the few subpowers algorithm and RETURN the answer
\ENDIF
\STATE Transform $\cP$ to a subdirectly irreducible, block-minimal and 
(2,3)-minimal instance
\STATE $\mu^*_v=\mu_v$ for $v\in\Razm(\cP)\cap\Centr(\cP)$ and 
$\mu^*_v=\zz_v$ otherwise
\STATE $\cP^*=\cP\fac{\ov\mu^*}$
\STATE /* the global 1-minimality of $\cP^*$
\FOR{every $v\in V$ and $a\in\zA_v\fac{\mu^*_v}$}
\STATE $\cP'=\cP^*_{(v,a)}$ \ \ \ \ /* Add constraint $\ang{(v),\{a\}}$
fixing the value of $v$ to $a$
\STATE Transform $\cP'$ to a subdirectly irreducible, (2,3)-minimal instance
$\cP''$
\STATE If $\razm(\cP'')<\razm(\cP)$ call {\sf SolveCSP} on $\cP''$ and 
flag $a$ if $\cP''$ has no solution
\STATE Establish block-minimality of $\cP''$; if the problem changes,
return to Step~10
\STATE If the resulting instance is empty, flag the element $a$ 
\ENDFOR
\STATE If there are flagged values, tighten the instance by removing the 
flagged elements and start over
\STATE Use Theorem~\ref{the:central} to reduce $\cP$ to an instance $\cP^\dg$ 
with $\razm(\cP^\dg)<\razm(\cP)$
\STATE Call {\sf SolveCSP} on $\cP^\dg$ and RETURN the answer
\end{algorithmic}
\end{algorithm}

We now illustrate the algorithm on our running example.

\begin{example}\label{exa:more7-small}
We illustrate the algorithm {\sf SolveCSP} on the instance from 
Example~\ref{exa:more4-small}. Recall that the domain of each variable is
$\zA_M$, its monolith is $\th$, and $(\zz:\th)$ is the full relation.
This means that $\razm(\cP)=3$, $\Razm(\cP)=V$ and $\Centr(\cP)=V$,
as well. Therefore we are in the case of large centralizers. Set
$\mu^*_{v_i}=\th_i$ for each $i\in[5]$ and consider the problem 
$\cP\fac{\ov\mu*}=(V,\{C^*_1=\ang{\bs_1,\rel^*_1}, 
C^*_2=\ang{\bs_2,\rel^*_2})$, where
$$
\rel^*=\left(\begin{array}{ccc}
0\fac\th & 2\fac\th & 2\fac\th\\ 0\fac\th & 2\fac\th & 2\fac\th\\ 
0\fac\th & 0\fac\th & 2\fac\th
\end{array}\right).
$$
It is an easy exercise to show that this instance is globally 1-minimal (every
value $0\fac\th$ can be extended to the all-$0\fac\th$ solution, and every 
value $2\fac\th$ can be extended to the all-$2\fac\th$ solution). This completes 
\emph{Stage 1}.
For every variable $v_i$ we choose $b\in\zA_M\fac\th$ such that for some
$a\in\zA_M\fac\th$ the pair $\{a,b\}$ is a semilattice edge. Since 
$\zA_M\fac\th$ is a 2-element semilattice, setting $b=0\fac\th$ and $a=2\fac\th$
is the only choice. Therefore $\vf_{v_i,b_i}$ in our case can be chosen to
be the same solution $\vf$ given by 
$\vf(v_i)=0\fac\th$; and \emph{Stage 2} is completed. For \emph{Stage 3}
first note that in $\zA_M$ the operation $r$ plays the role of multiplication $\cdot$. 
Then for each of the constraints $C^1,C^2$ choose a representative 
$\ba_1\in\rel_1\cap(\vf(v_1)\tm\vf(v_2)\tm\vf(v_3))=\rel_1\cap\{0,1\}^3$, 
$\ba_2\in\rel_2\cap(\vf(v_2)\tm\vf(v_4)\tm\vf(v_5))=\rel_2\cap\{0,1\}^3$,
and set $\cP'=(\{v_1\zd v_5\},\{C'_1=\ang{(v_1,v_2,v_3),\rel'_1}, 
C'_2=\ang{(v_2,v_4,v_5),\rel'_2}\})$, where $\rel'_1=r(\rel_1,\ba)$,
$\rel'_2=r(\rel_2,\bb)$. Since $r(2,0)=r(2,1)=0$, regardless of the choice 
of $\ba,\bb$ in our case $\rel'_1\sse\rel_1,\rel'_2\sse\rel_2$, and are 
invariant with respect to the affine operation of $\zZ_2$. Therefore the 
instance $\cP'$ can be viewed as a system of linear equations over $\zZ_2$
(this system is actually empty in our case), and can be easily solved.
\EEX
\end{example}

Using Lemma~\ref{lem:to-block-minimality} and 
Theorems~\ref{the:non-central},\ref{the:central} it is not difficult to see 
that the algorithm runs in polynomial time. 

\begin{theorem}\label{the:algorithm}
Algorithm SolveCSP (Algorithm~\ref{alg:csp}) correctly solves every instance 
from $\CSP(\cA)$ and runs in polynomial time.
\end{theorem}

\begin{proof}
By the results of \cite{Bulatov16:restricted,Bulatov20:restricted} the algorithm correctly solves
the given instance $\cP$ in polynomial time if the conditions of Step~1 are 
true. Lemma~\ref{lem:to-block-minimality} implies
that Steps~4 and~12 can be completed by recursing to strictly smaller 
instances. 

Next we show that the for-loop in Steps 8-14 checks if 
$\cP^*=\cP\fac{\ov\mu^*}$ is globally 1-minimal. For this we need to 
verify that a value $a$ is flagged if and only if $\cP^*$ has no solution 
$\vf$ with $\vf(v)=a$, and therefore if no values are flagged then 
$\cP^*$ is globally 1-minimal.
If $\vf(v)=a$ for some solution $\vf$ of $\cP^*$, then $\vf$
is a solution $\cP'$ constructed in Step~9. In this case Steps~11,12 cannot
result in an empty instance. 
Suppose $a\in\zA_v\fac{\mu^*_v}$ is not flagged. If
$\razm(\cP'')<\razm(\cP)$ this means that $\cP''$ and therefore $\cP'$
has a solution. Otherwise this means that establishing block-minimality of 
$\cP''$ is successful. In this case $\cP''$ has a solution by 
Theorem~\ref{the:non-central}, because $\Razm(\cP'')\cap\Centr(\cP'')=\eps$. 
This in turn implies that $\cP'$ has a solution.  
Observe also that the set of unflagged values for each variable $v\in V$
is a subalgebra of $\zA\fac{\mu^*}$. Indeed, the set of solutions of $\cP^*$
is a subalgebra $\cS^*$ of $\prod_{v\in V}\zA\fac{\mu^*}$, and the set of 
unflagged values is the projection of $\cS^*$ on the coordinate position $v$.

Finally, if Steps~8--15 are completed without restarts, Steps~16,17 
can be completed by Theorem~\ref{the:central}, 
and recursing on $\cP'$ such that either $\razm(\cP')<\razm(\cP)$ or
$\Razm(\cP')\cap\Centr(\cP')=\eps$.

To see that the algorithm runs in polynomial time it suffices to observe 
that\\[1mm] 
(1) The number of restarts in Steps~4 and~15 is at most linear, as the 
instance becomes smaller after every restart; therefore the number of
times Steps~4--15 are executed together is at most linear.\\[1mm]
(2) The number of iterations of the for-loop in Steps 8--14 is linear.\\[1mm]
(3) The number of restarts in Steps~10 and~12 is at most linear, as the
instance becomes smaller after every iteration.\\[1mm]
(4) Every call of SolveCSP when establishing block-minimality in Steps~4, 
and~12 is made on an instance strictly smaller than $\cP$, and therefore
the depth of recursion is bounded by $\razm(\cP)$ in Step~4,11,12 
and~17. \\[2mm]
Thus a more thorough estimation gives a bound on the running time of 
$O(n^{3k})$, where $k$ is the maximal size of an algebra in $\cA$.
\end{proof}

\subsection{Proof of Theorem~\ref{the:central}}\label{sec:theorem-47}

Following \cite{Maroti10:tree} let $\cP=(V,\cC)$ be an instance and 
$p_v\colon\zA_v\to\zA_v$, $v\in V$. 
Mappings $p_v$, $v\in V$, are said to be 
\emph{consistent} if for any
$\ang{\bs,\rel}\in\cC$, $\bs=(\vc vk)$, and any tuple $\ba\in\rel$ the
tuple $(p_{v_1}(\ba[1])\zd p_{v_k}(\ba[k]))$ belongs to $\rel$. 
It is easy to see that the composition of two families of consistent mappings 
is also a consistent mapping. 
For consistent idempotent mappings $p_v$ by $p(\cP)$ 
we denote the \emph{retraction} of $\cP$, that is, $\cP$ 
restricted to the images of $p_v$. In this case $\cP$ has a solution if and only if 
$p(\cP)$ has, see \cite{Maroti10:tree}. 

Let $\vf$ be a solution of $\cP\fac{\ov\mu^*}$. We define 
$p^\vf_v:\zA_v\to\zA_v$ as follows: $p^\vf_v=q_v^k$, where 
$q_v(a)=a\cdot b_v$, element $b_v$ is any 
element of $\vf(v)$, and $k$ is such that $q_v^k$ is idempotent for all $v\in V$. 
Note that by Lemma~\ref{lem:centralizer-multiplication} this mapping is 
properly defined even if $\mu^*_v\ne\zz_v$.

\begin{lemma}\label{lem:consistent-mapping}
Mappings $p^\vf_v$, $v\in V$, are consistent. 
\end{lemma}

\begin{proof}
Take any $C=\ang{\bs,\rel}\in\cC$. Since $\vf$ is a solution of 
$\cP\fac{\ov\mu^*}$, there is $\bb\in\rel$ such that $\bb[v]\in\vf(v)$ 
for $v\in\bs$. Then for any $\ba\in\rel$, 
$q(\ba)=\ba\cdot\bb\in\rel$, and this product does not depend on the choice
of $\bb$, as it follows from Lemma~\ref{lem:centralizer-multiplication}. 
Iterating this operation also produces a tuple from $\rel$.
\end{proof}

We are now in a position to prove Theorem~\ref{the:central}.

\begin{proof}[of Theorem~\ref{the:central}]
We need to show 3 properties of the problem $\cP^\dg$ constructed in Stage~3:
(a) $\cP$ has a solution if and only if $\cP^\dg$ does; (b) for every $v\in\Razm(\cP)$,
$|\zA^\dg_v|<|\zA_v|$, where $\zA^\dg_v$ is the domain of $v$ in $\cP^\dg$; and
(c) every algebra $\zA^\dg_v$ has a weak near-unanimity term operation. We use the 
inductive definition of $\cP^\dg$ given in Stage~3.

Recall that $\Razm(\cP)=\{\vc v\ell\}$, $a_i,b_i\in\zA_{v_i}$ are such that 
$a_i\le b_i$ and $b_i\in\vf_{v_i,b_i}(v_i)$, where $\vf_{v_i,b_i}$ is a 
solution of $\cP\fac{\ov\mu^*}$. For $v\in V$ let mapping $p_{vi}:\zA_v\to\zA_v$
be given by 
\[
p_{vi}(x)=(\dots(x\cdot\vf_{v_1,b_1}(v))\cdot\dots)\cdot\vf_{v_i,b_i}(v),
\]
where if $\mu^*_v=\mu_v$ by Lemma~\ref{lem:centralizer-multiplication} 
the multiplication by $\vf_{v_j,b_j}(v)$ does not depend on the choice of 
a representative from $\vf_{v_j,b_j}(v)$. 
By Lemma~\ref{lem:consistent-mapping} $\{p_{vi}\}$ for every $i$, and so 
$\{p_v\}$ and $\{p^k_v\}$ are collections of consistent mappings. Now (a) 
follows from \cite{Maroti10:tree}.

Next we show that for every $j\le i\le\ell$ it holds that 
$|p_{vi}(\zA_{v_j})|<|\zA_{v_j}|$. Since applying mappings to a set does not
increase its cardinality, this implies (b). If $|p_{vj-1}(\zA_{v_j})|<|\zA_{v_j}|$,
we have the desired inequality applying the observation in the previous sentence.
Otherwise $a_j\in \zA_{v_j}=p_{vj-1}(\zA_{v_j})$, and it suffices to notice that 
$a_j\cdot\vf_{v_j,b_j}(v_j)=b_j\cdot\vf_{v_j,b_j}(v_j)=b_j$.

To prove (c) observe that if $\zA_v$ is semilattice free then $p^\vf_v$ is 
the identity mapping for any $\vf$ by Lemma~\ref{lem:multiplication-p1}, 
and so $\zA^\dg_v=\zA_v$. For the remaining domains let 
$f$ be a weak near-unanimity term of the class $\cA$. Then for any
idempotent mapping $p$ the operation $p\circ f$ given by
$(p\circ f)(\vc xn)=p(f(x_1\zd x_n))$
is a weak near-unanimity term of $p(\cA)=\{p(\zA)\mid \zA\in\cA\}$.
The result follows.
\end{proof}

\section{Algebra technicalities}\label{sec:compressed}

The rest of the paper is dedicated to proving Theorem~\ref{the:non-central}. 
This part assumes some familiarity with algebraic terminology. A brief review 
of the necessary facts from universal algebra can be found in 
\cite{Bulatov20:algebraic}. In this section we remind some results from 
\cite{Bulatov20:algebraic} necessary for our proof.

\subsection{Coloured graphs}

In \cite{Bulatov04:graph,Bulatov08:recent} we introduced a local approach 
to the structure of finite algebras. As we use this approach in the proof of 
Theorem~\ref{the:non-central}, we present the necessary elements of it here, 
see also \cite{Bulatov20:graph,Bulatov20:maximal}. 
For the sake of the definitions below we slightly abuse terminology 
and by a module mean the full idempotent reduct of a module.

For an algebra $\zA$ the graph $\cG(\zA)$ is defined as follows. 
The vertex set is the universe $A$ of $\zA$. A pair $ab$ of vertices is an 
\emph{edge} if and only if there exists a maximal congruence $\th$ of 
$\Sg{a,b}$, and a term operation $f$ of $\zA$ 
such that either $\Sg{a,b}\fac\th$ is a module and $f$ is an affine 
operation on it, or $f$ is a semilattice operation on
$\{a\fac\th,b\fac\th\}$, or $f$ is a majority operation on
$\{a\fac\th,b\fac\th\}$. (Note that we use the same operation symbol in this case.)
If there are a maximal congruence $\th$ and a term operation $f$ of $\zA$ such that $f$ 
is a semilattice operation on $\{a\fac\th,b\fac\th\}$ then $ab$ is said to have the
\emph{semilattice type}. An edge $ab$ is of 
\emph{majority type} if there are a maximal congruence $\th$ and a term 
operation $f$ such that $f$ is a majority
operation on $\{a\fac\th,b\fac\th\}$ and there is no semilattice 
term operation on $\{a\fac\th,b\fac\th\}$. Finally, $ab$ 
has the \emph{affine type} if there are $\th$ and $f$ 
such that $f$ is an affine operation on $\Sg{a,b}\fac\th$ and 
$\Sg{a,b}\fac\th$ is a module. Pairs of the form $\{a\fac\th,b\fac\th\}$
will be referred to as \emph{thick edges}.

Properties of $\cG(\zA)$ are related to the properties of the algebra $\zA$.

\begin{theorem}[Theorem~5 of \cite{Bulatov20:graph}]%
\label{the:connectedness}
Let $\zA$ be an idempotent algebra $\zA$ such that 
$\var\zA$ omits type \one. Then
\begin{itemize}
\item[(1)]
any two elements of $\zA$ are connected by a sequence of edges of
the semilattice, majority, and affine types;
\item[(2)]
$\var\zA$ omits types \one\ and \two\ if and only if $\cG(\zA)$ satisfies
the conditions of item (1) and contains no edges of the affine type.
\end{itemize}
\end{theorem}

We use the following refinement of this construction. Let $\cA$ be a finite class of
finite smooth algebras. A ternary term operation $g'$ of $\cA$ is said to satisfy the 
\emph{majority condition} for $\cA$ if $g'$ is a majority operation 
on every thick majority edge of every algebra from $\cA$. A ternary term operation 
$h'$ is said to satisfy the 
\emph{minority condition} for $\cA$ if $h'$ is a Mal'tsev operation on every thick affine
edge. Operations satisfying the majority and minority conditions always exists, 
as is proved in \cite[Theorem~21]{Bulatov20:graph}. Fix an operation $h$ 
satisfying the minority condition, it can also be chosen to satisfy the equation 
$h(h(x,y,y),y,y)=h(x,y,y)$. A pair of elements $a,b\in\zA\in\cA$ is said to be 
\begin{itemize}
\item[(1)]
a \emph{semilattice edge} if there is a term operation $f$ such that 
$f(a,b)=f(b,a)=b$;
\item[(2)]
a \emph{thin majority edge} if for any term operation $g'$ satisfying the 
majority condition the subalgebras $\Sg{a,g'(a,b,b)},\Sg{a,g'(b,a,b)},
\Sg{a,g'(b,b,a)}$ contain $b$.
\item[(3)]
a \emph{thin affine edge} if $h(b,a,a)=b$ and $b\in\Sg{a,h'(a,a,b)}$ for any 
term operation $h'$ satisfying the minority condition.
\end{itemize}
Note that thin edges are directed, as $a$ and $b$ appear asymmetrically. 
By $\cG'(\zA)$ we denote the graph whose vertices are the elements of 
$\zA$, and the edges are the thin edges defined above. 
Theorem~21 from \cite{Bulatov20:graph} also implies that there exists
a binary term operation $\cdot$ of $\zA$ that is a semilattice operation 
on every thin semilattice edge. 

We distinguish several types of paths in $\cG'(\zA)$ depending on the 
types of edges involved. A directed path in $\cG'(\zA)$ is called an 
\emph{asm-path},  if there is an asm-path from $a$ to $b$ we write 
$a\sqq_{asm} b$. If all edges of this path 
are semilattice or affine, it is called an \emph{affine-semilattice path} or 
an \emph{as-path}, if there is an as-path from $a$ to $b$ we write 
$a\sqq_{as} b$. We consider strongly connected components 
of $\cG'(\zA)$ with majority edges 
removed, and the natural partial order on such components. The maximal 
components will be called \emph{as-components}, and 
the elements from as-components are called 
\emph{as-maximal}; the set of all
as-maximal elements of $\zA$ is denoted by $\amax(\zA)$. 
An alternative way to define as-maximal elements is as 
follows: $a$ is as-maximal if for every $b\in\zA$ such that $a\sqq_{as} b$ it also 
holds that $b\sqq_{as}a$. Finally, element $a\in\zA$ is said to be 
\emph{universally maximal} (or \emph{u-maximal} for short) 
if for every $b\in\zA$ such that $a\sqq_{asm} b$ 
it also holds that $b\sqq_{asm}a$. The set of all u-maximal elements of 
$\zA$ is denoted $\umax(\zA)$.

U-maximality has additional useful properties.

\begin{lemma}[Theorem~23, \cite{Bulatov20:maximal}; %
Lemma~12, \cite{Bulatov20:algebraic}]%
\label{lem:u-max-congruence}
(1) Any two u-maximal elements are connected with an asm-path,\\[2mm]
(2) Let $\zB$ be a subalgebra of $\zA$ containing a u-maximal element of $\zA$. 
Then every element u-maximal in $\zB$ is also 
u-maximal in $\zA$. In particular, if $\al$ is a congruence of $\zA$ and $\zB$ 
is a u-maximal $\al$-block, that is $\zB$ is a u-maximal element in 
$\zA\fac\al$, then $\umax(\zB)\sse\umax(\zA)$. 
\end{lemma}

Relations, or, more generally subdirect products of algebras can be naturally
endowed with a graph structure: Let $\rel$ be a subdirect product of 
$\zA_1\tms\zA_n$. A pair $\ba,\bb\in\rel$ is a thin \{semilattice, majority, 
affine\} edge if for every $i\in[n]$ the pair $\ba[i],\bb[i]$ is a 
thin \{semilattice, majority, affine\} edge or $\ba[i]=\bb[i]$ (in the latter 
case it will often be convenient to call a pair of equal elements a thin edge of 
whatever type we need). Paths and maximality can also be 
lifted to subdirect products.

\begin{lemma}[The Maximality Lemma, Corollaries~18,19, 
\cite{Bulatov20:maximal}]\label{lem:to-max}
Let $\rel$ be a subdirect product of $\zA_1\tms\zA_n$, $I\sse[n]$.\\[1mm]
(1) For any $\ba\in\rel$, and an as-path (asm-path) 
$\vc\bb k\in\pr_I\rel$ with $\pr_I\ba=\bb_1$, there is an
as-path (asm-path) $\vc{\bb'}\ell\in\rel$ 
such that $\pr_I\bb'_\ell=\bb_\ell$.\\[1mm]
(2) For any $\bb\in\amax(\pr_I\rel)$ ($\bb\in\umax(\pr_I\rel)$) there 
is $\bb'\in\amax(\rel)$ ($\bb'\in\umax(\rel)$), such that 
$\pr_I\bb'=\bb$. \\[1mm]
(3) If $\ba\in\rel$ is a as-maximal or u-maximal element then so is 
$\pr_I\ba$.
\end{lemma}

We complete this section with an auxiliary statement that will be needed
later.

\begin{lemma}[Lemmas~15, \cite{Bulatov20:algebraic}, Lemma~4.14, 
\cite{Hobby88:structure}]\label{lem:as-type-2}
(1) Let $\al\prec\beta$, $\al,\beta\in\Con(\zA)$, let $B$ be a $\beta$-block and 
$\typ(\al,\beta)=\two$. Then $B\fac\al$ is term equivalent to a module. 
In particular, every pair of elements of $B\fac\al$ is a thin affine edge in 
$\zA\fac\al$.\\[2mm]
(2) If $(\al:\beta)\ge\beta$, then $\typ(\al,\beta)=\two$.
\end{lemma}

\subsection{Quasi-decomposition and rectangularity}

We make use of the property of quasi-2-decomposability proved 
in \cite{Bulatov20:maximal}.

\begin{theorem}[The 2-Decomposition Theorem~30, \cite{Bulatov20:maximal}]%
\label{the:quasi-2-decomp}
If $\rel$ is an $n$-ary relation, $X\sse[n]$, tuple $\ba$ is such that
$\pr_J\ba\in\pr_J\rel$ for any $J\sse[n]$, $|J|=2$, and $\pr_X\ba
\in\amax(\pr_X\rel)$, there is a tuple $\bb\in\rel$ with
$\pr_J\ba\sqq_{as}\pr_J\bb$ 
for any $J\sse[n]$, $|J|=2$, and $\pr_X\bb=\pr_X\ba$. 
\end{theorem}

Another property of relations was also introduced in \cite{Bulatov20:maximal} 
and is similar to the rectangularity property of relations with a Mal'tsev 
polymorphism. Let $\rel$ be a subdirect product of $\zA_1,\zA_2$. 
By $\lnk_1,\lnk_2$ we denote the congruences of
$\zA_1,\zA_2$, respectively, generated by the sets of pairs
$\{(a,b)\in\zA_1^2\mid \text{ there is $c\in\zA_2$ such that } 
(a,c),(b,c)\in\rel\}$ and $\{(a,b)\in\zA_2^2\mid \text{ there is 
$c\in\zA_1$ such that } (c,a),(c,b)\in\rel\}$, respectively. Congruences 
$\lnk_1,\lnk_2$ are called \emph{link congruences}.
Relation $\rel$ is said to be \emph{linked} 
if the link congruences are full congruences.

\begin{prop}[Corollary~28, \cite{Bulatov20:maximal}]\label{pro:max-gen} 
Let $\rel$ be a subdirect product of $\zA_1$ and $\zA_2$, 
$\lnk_1,\lnk_2$ the link congruences, and let $B_1,B_2$ be 
as-components of a $\lnk_1$-block and a $\lnk_2$-block, respectively, 
such that $\rel\cap(B_1\tm B_2)\ne\eps$. Then $B_1\tm B_2\sse\rel$.

In particular, if $\rel$ is linked and $B_1,B_2$ are as-components of 
$\zA_1,\zA_2$, respectively, such that $\rel\cap(B_1\tm B_2)\ne\eps$, 
then $B_1\tm B_2\sse\rel$.
\end{prop}

\subsection{Separating congruences}\label{sec:separating-congruences}

Let $\zA$ be a finite algebra and $\al,\beta\in\Con(\zA)$. The pair $\al,\beta$
is said to be a \emph{prime interval}, denoted $\al\prec\beta$ if $\al<\beta$
and for any $\gm\in\Con(\zA)$ with $\al\le\gm\le\beta$ either $\al=\gm$
or $\beta=\gm$. For $\al\prec\beta$, an \emph{$(\al,\beta)$-minimal set} is a 
set minimal with respect to inclusion among the sets of the form $f(\zA)$, where 
$f$ is a unary polynomial of $\zA$ such that $f(\beta)\not\sse\al$. 

For an $(\al,\beta)$-minimal set $U$ and a $\beta$-block $B$ such that 
$\beta\red{U\cap B}\ne\al\red{U\cap B}$, the set $U\cap B$ is said
to be an \emph{$(\al,\beta)$-trace}. A 2-element set 
$\{a,b\}\sse U\cap B$ such that $(a,b)\in\beta-\al$, is called an 
\emph{$(\al,\beta)$-subtrace}. 

Let $\al\prec\beta$ and $\gm\prec\dl$ be prime
intervals in $\Con(\zA)$. We say that $(\al,\beta)$ can be 
\emph{separated} from 
$(\gm,\dl$  if there is a unary polynomial $f$ of $\zA$ such that 
$f(\beta)\not\sse\al$, but $f(\dl)\sse\gm$. The polynomial $f$ in this case is 
said to \emph{separate} $(\al,\beta)$ from $(\gm,\dl)$. 

In a similar way separation can be defined for prime intervals in different 
coordinate positions of a relation. Let $\rel$ be a subdirect product of 
$\zA_1\tm\dots\tm\zA_n$. Then $\rel$ is also an algebra and its polynomials 
can be defined in the same way as for a single algebra. Let $i,j\in[n]$ and let
$\al\prec\beta$, $\gm\prec\dl$ be prime intervals in $\Con(\zA_i)$ and 
$\Con(\zA_j)$, respectively. Interval $(\al,\beta)$ can be separated from 
$(\gm,\dl)$ if there is a unary polynomial $f$ of $\rel$ such that 
$f(\beta)\not\sse\al$ but $f(\dl)\sse\gm$ (note that the actions of $f$ on 
$\zA_i,\zA_j$ are polynomials of those algebras). 

If $\vc\zA n$ are algebras and $\vc Bn$ are their subsets $B_i\sse\zA_i$,
$i\in[n]$, and $\vc\al n$ are congruences of the $\zA_i$'s, it will be 
convenient to denote $B_1\tms B_n$ by $\ov B$ and $\beta_1\tms\beta_n=
\{(\ba,\bb)\in(\zA_1\tms\zA_n)^2\mid \ba[i]\eqc{\al_i}\bb[i], i\in[n]\}$
by $\ov\beta$. By $\Cgg\zA D$, or just $\Cg D$ if $\zA$ is clear from the 
context, we denote the congruence of $\zA$ generated by a set $D$ of
pairs from $\zA^2$.

For an algebra $\zA$, a set $\cU$ of unary polynomials, and $B\sse\zA^2$,
we denote by $\Cgg{\zA,\cU}B$ the transitive-symmetric closure
of the set $T(B,\cU)=\{(f(a),f(b))\mid (a,b)\in B, f\in\cU\}$. Let also 
$\al,\beta\in\Con(\zA)$, $\al\le\beta$ and $D$ a subuniverse of $\zA$ such that 
$\beta=\Cgg\zA{\al\cup\{(a,b)\}}$ for some $a,b\in D$. We say that $\al$ 
and $\beta$ are \emph{$\cU$-chained} with respect to $D$ if for any 
$\beta$-block $B$ such that $B'=B\cap\umax(D)\ne\eps$ we have 
$(\umax(B'))^2\sse\Cgg{\zA,\cU}{\al\cup\{(a,b)\}}$. 

Let $\beta_i\in\Con(\zA_i)$, let $B_i$ be a $\beta_i$-block for $i\in[n]$,
and let $\rel'=\rel\cap\ov B$, $B'_i=\pr_i\rel'$. A unary polynomial $f$ 
is said to be \emph{$\ov B$-preserving} if $f(\ov B)\sse\ov B$. 
We call an $n$-ary relation $\rel$ \emph{chained} 
with respect to $\ov\beta,\ov B$ if\\[2mm]
(Q1) for any $I\sse[n]$ and $\al,\beta\in\Con(\pr_I\rel)$ such that 
$\al\le\beta\le\ov\beta_I$, $\al,\beta$ are $\cU_B$-chained with respect to 
$\pr_I\rel'$, and $\cU_B$ is the set of all 
$\ov B$-preserving polynomials of $\rel$;\\[2mm]
(Q2) for any $\al,\beta\in\Con(\pr_I\rel)$, $\gm,\dl\in\Con(\zA_j)$, 
$j\in[n]$, such that $\al\prec\beta\le\ov\beta_I$, $\gm\prec\dl\le\beta_j$, 
and $(\al,\beta)$ can be separated from $(\gm,\dl)$, the congruences 
$\al$ and $\beta$ are $\cU(\gm,\dl,\ov B)$-chained with respect to 
$\pr_I\rel'$, where $\cU(\gm,\dl,\ov B)$ 
is the set of all $\ov B$-preserving polynomials $g$ of $\rel$ such that 
$g(\dl)\sse\gm$.

The following lemma claims that the property to be 
chained is preserved 
under certain transformations of $\ov\beta$ and $\ov B$.

\begin{lemma}[Lemmas~44,45, \cite{Bulatov20:algebraic}]%
\label{lem:S7}
Let $\rel$ be a subdirect product of $\vc\zA n$.\\[2mm]
(1) Let $\beta_i=\zo_{\zA_i}$ and $B_i=\zA_i$ for $i\in[n]$. Then 
$\rel$ is chained with respect to $\ov\beta,\ov B$.\\[2mm]
(2) Let $\beta_i\in\Con(\zA_i)$ and $B_i$ a $\beta_i$-block, $i\in[n]$, 
be such that $\rel$ is chained
with respect to $\ov\beta,\ov B$. Let $\rel'=\rel\cap\ov B$
and $B'_i=\pr_i\rel'$. Fix $i\in[n]$, $\beta'_i\prec\beta_i$,
and let $D_i$ be a $\beta'_i$-block that is as-maximal in 
$B'_i\fac{\beta'_i}$. Let also $\beta'_j=\beta_j$ and $D_j=B_j$ 
for $j\ne i$. Then $\rel$ is chained with respect to 
$\ov\beta',\ov D$.
\end{lemma}

Let again $\rel$ be a subdirect product of $\zA_1\tms\zA_n$ and let 
$\cW^\rel$ denote the set of triples $(i,\al,\beta)$, where 
$i\in[n]$ and $\al,\beta\in\Con(\zA_i)$, $\al\prec\beta$. 
We say that
$(i,\al,\beta)$ cannot be separated from $(j,\gm,\dl)$ if $(\al,\beta)$
cannot be separated from $(\gm,\dl)$ in $\rel$. Then the 
relation `cannot be separated' on $\cW^\rel$ is clearly reflexive 
and transitive. The next lemma shows that it is to some extent
symmetric.

\begin{lemma}[Theorem~30, \cite{Bulatov20:algebraic}]%
\label{lem:relative-symmetry}
Let $\rel$ be a subdirect product of $\zA_1\tms\zA_n$, for each $i\in[n]$, 
$\beta_i\in\Con(\zA_i)$, $B_i$ a $\beta_i$-block such that $\rel$ is  
chained with respect to $\ov\beta,\ov B$; $\rel'=\rel\cap\ov B$, 
$B'_i=\pr_i\rel'$. Let also $\al\prec\beta\le\beta_1$, $\gm\prec\dl=\beta_2$, 
where $\al,\beta\in\Con(\zA_1)$, $\gm,\dl\in\Con(\zA_2)$. 
If $B'_2\fac\gm$  has a nontrivial as-component $D$ and $(\al,\beta)$ 
can be separated from $(\gm,\dl)$,
then there is a $\ov B$-preserving polynomial $g$ such that 
$g(\beta\red{B'_1})\sse\al$ and $g(\dl)\not\sse\gm$. Moreover, for any
$c,d\in D$ polynomial $f$ can be chosen such that $f(c)=c,f(d)=d$.
\end{lemma}

We also introduce polynomials that 
collapse all prime intervals in congruence lattices of factors of a subproduct,
except for a set of intervals that cannot be separated from each other.

Let $\rel$ be a subdirect product of $\zA_1\tms\zA_n$, and choose 
$\beta_j\in\Con(\zA_j)$, $j\in[n]$. Let also $i\in[n]$, and 
$\al,\beta\in\Con(\zA_i)$ be such that $\al\prec\beta\le\beta_i$; 
let also $B_j$ be a $\beta_j$-block, $j\in[n]$. We call an idempotent 
unary polynomial $f$ of $\rel$ \emph{$\al\beta$-collapsing for 
$\ov\beta,\ov B$} if 
\begin{itemize}
\item[(a)]
$f$ is $\ov B$-preserving;
\item[(b)]
$f(\zA_i)$ is an $(\al,\beta)$-minimal set, in particular $f(\beta)\not\sse\al$; 
\item[(c)]
$f(\dl\red{B_j})\sse\gm\red{B_j}$ for every 
$\gm,\dl\in\Con(\zA_j)$, $j\in[n]$, with $\gm\prec\dl\le\beta_j$, and such 
that $(\al,\beta)$ can be separated from $(\gm,\dl)$ or $(\gm,\dl)$
can be separated from $(\al,\beta)$.
\end{itemize}

\begin{lemma}[Theorem~40, \cite{Bulatov20:algebraic}]%
\label{lem:collapsing}
Let $\rel$, $i$, $\al,\beta$, and $\beta_j$, $j\in[n]$, be as above and $\rel$ 
chained with respect to $\ov\beta,\ov B$. Let also $\rel'=\rel\cap\ov B$.
Then if $\beta=\beta_i$ and $\pr_i\rel'\fac\al$ contains a 
nontrivial as-component, then there exists an $\al\beta$-collapsing 
polynomial $f$ for $\ov\beta,\ov B$. Moreover, $f$ can be chosen to
satisfy any one of the following conditions:\\[1mm]
(d) for any $(\al,\beta)$-subtrace $\{a,b\}\sse\amax(\pr_i\rel')$ with $b\in\as(a)$,
polynomial $f$ can be chosen such that $a,b\in f(\zA_i)$;\\[1mm]
(e) if $\typ(\al,\beta)=\two$, for any $\ba\in\umax(\rel')$ 
polynomial $f$ can be chosen such that $f(\ba)=\ba$;\\[1mm]
(f) if $\typ(\al,\beta)=\two$, $\ba\in\umax(\rel'')$, where 
$\rel''=\{\bb\in\rel\mid \bb[i]\eqc\al\ba[i]\}$ and 
$\{a,b\}\sse\amax(\pr_i\rel')$ is an $(\al,\beta)$-subtrace such 
that $\ba[i]=a$ and $b\in\as(a)$, then polynomial $f$ can be chosen 
such that $f(\ba)=\ba$ and $a,b'\in f(\zA_i)$ for some $b'\eqc\al b$.
\end{lemma}

\subsection{The Congruence Lemma}\label{sec:congruence}

This section contains a technical result, the Congruence
Lemma~\ref{lem:affine-link}, that will be used when 
proving Theorem~\ref{the:non-central}. We start with 
introducing two closure properties of algebras and their subdirect products. 
Although we do not need as-closeness right now, it fits well with polynomial 
closeness.

Let $\rel$ be a subdirect product of $\vc\zA n$ and $\relo$ a subalgebra of 
$\rel$. We say that $\relo$ is \emph{polynomially closed} in $\rel$ if 
for any polynomial $f$ of $\rel$ the following condition
holds: for any $\ba,\bb\in\umax(\relo)$ such that $f(\ba)=\ba$ and 
for any $\bc\in\Sg{\ba,f(\bb)}$ such that 
$\ba\sqq_{as}\bc$ in $\Sg{\ba,f(\bb)}$, the tuple $\bc$ belongs 
to $\relo$. A subset $\rela\sse\relo$ is \emph{as-closed} in $\relo$ if for 
any $\ba,\bb\in\relo$ with $\ba\in\umax(\rela)$, $\ba\sqq_{as}\bb$ in 
$\relo$, it holds $\bb\in\rela$. The set $\rela$ is said to be \emph{weakly 
as-closed} in $\relo$ if for any $i\in[n]$, $\pr_i\rela$ is as-closed in 
$\pr_i\relo$.

Polynomially closed subalgebras and as-closed subsets are well behaved 
with respect to some standard algebraic transformations.

\begin{lemma}[Lemma~42,  \cite{Bulatov20:algebraic}]%
\label{lem:poly-closed}
(1) For any $\rel$, $\rel$ is polynomially closed in $\rel$ and $\rel$
is as-closed in $\rel$.\\[1mm]
(2) Let $\relo_i$ be polynomially closed in $\rel_i$, $i\in[k]$, and let
$\rel,\relo$ be pp-defined through $\vc\rel k$ and $\vc\relo k$, respectively, 
by the same pp-formula $\exists\ov x\Phi$; that is, 
$\rel=\exists\ov x\Phi(\vc\rel k)$ and $\relo=\exists\ov x\Phi(\vc\relo k)$. 
Let also $\rel'=\Phi(\vc\rel k)$ and $\relo'=\Phi(\vc\relo k)$, and suppose that 
for every atom $\rel_i(\vc x\ell)$ and any 
$\ba\in\umax(\rel_i)$ there is $\bb\in\rel'$ with $\pr_{\{\vc x\ell\}}\bb=\ba$,
and also $\umax(\relo')\cap\umax(\rel')\ne\eps$. 
Then $\relo$ is polynomially closed in $\rel$.

If also $\rela_i\sse\relo_i$ are as-closed in $\relo_i$, then 
$\rela=\Phi(\vc\rela k)$ is as-closed in $\relo$.\\[1mm]
(3) Let $\rel$ be a subdirect product of $\vc\zA n$, $\beta_i\in\Con(\zA_i)$, 
$i\in[n]$, and let $\relo$ be polynomially closed in $\rel$. Then 
$\relo\fac{\ov\beta}$ is polynomially closed in $\rel$. 

If $\rela\sse\relo$ is as-closed in $\relo$ then $\rela\fac{\ov\beta}$ 
is as-closed in $\relo\fac{\ov\beta}$.
\end{lemma}

We are now in a position to state the Congruence Lemma.
Let $\rel$ be a subdirect product of $\zA_1\tm\zA_2$, 
$\beta_1,\beta_2$ congruences of $\zA_1,\zA_2$, and let $B_1,B_2$ be 
$\beta_1$- and $\beta_2$-blocks, respectively. Also, let $\rel$ be  
chained with respect to $(\beta_1,\beta_2),(B_1,B_2)$ and 
$\rel^*=\rel\cap(B_1\tm B_2)$,
$B^*_1=\pr_1\rel^*,B^*_2=\pr_2\rel^*$. Let $\al\in\Con(\zA_1)$ be 
such that $\al\prec\beta_1$. 

\begin{lemma}[The Congruence Lemma, Lemma~43,
\cite{Bulatov20:algebraic}]\label{lem:affine-link}
Suppose $\al=\zz_1$ and let $\rel'$ be a subalgebra of $\rel^*$ 
polynomially closed in $\rel$ and such that $B'_1=\pr_1\rel'$ contains 
an as-component $C$ of $B^*_1$ and $\rel'\cap\umax(\rel^*)\ne\eps$.
Let $\beta'$ be the least congruence of $\zA_2$ such that $\umax(B''_2)$, 
where $B''_2=\rel'[C]$ is a subset of a $\beta'$-block. Then either \\[2mm]
(1) $C\tm\umax(B''_2)\sse\rel'$, or\\[1mm] 
(2) there is $\eta\in\Con(\zA_2)$ with $\eta\prec\beta'\le\beta_2$ 
such that the intervals $(\al,\beta_1)$ and $(\eta,\beta')$ cannot 
be separated. \\
Moreover, in case (2) $\rel'\cap(C\tm B''_2)$ is the graph of a mapping 
$\vf: B''_2\to C$ such that the kernel of $\vf$ is the restriction 
of $\eta$ on $B''_2$.
\end{lemma}

\section{Decompositions and compressed problems}%
\label{sec:csp-non-central}

In this section we apply the machinery developed in the previous section 
to constraints satisfaction problems in order to prove 
Theorem~\ref{the:non-central}.

\subsection{Decomposition of CSPs}\label{sec:csp-decomposition}

We begin with showing how separating congruence intervals and
centralizers can be combined to obtain strands and therefore useful
decompositions of CSPs. The case of binary relations is settled in
\cite{Bulatov20:algebraic}.

\begin{lemma}[Lemma~34, \cite{Bulatov20:algebraic}]%
\label{lem:delta-alignment}
Let $\rel$ be a subdirect product of 
$\zA_1\tm\zA_2$, $\al_i,\beta_i\in\Con(\zA_i)$, $\al_i\prec\beta_i$, 
for $i=1,2$. If $(\al_1,\beta_1)$ and $(\al_2,\beta_2)$ cannot be separated 
from each other, then the coordinate positions 1,2 are 
$\zeta_1\zeta_2$-aligned in $\rel$, where 
$\zeta_1=(\al_1:\beta_1), \zeta_2=(\al_2:\beta_2)$.
\end{lemma}

Let $\cP=(V,\cC)$ be a (2,3)-minimal instance and let
$\ov\beta$, $\beta_v\in\Con(\zA_v)$, $v\in V$, be a collection of congruences. 
Let $\cW^\cP(\ov\beta)$ denote the set of triples $(v,\al,\beta)$ 
such that $v\in V$, $\al,\beta\in\Con(\zA_v)$, and $\al\prec\beta\le\beta_v$. 
Also, $\cW^\cP$ denotes $\cW^\cP(\ov\beta)$ when $\beta_v=\zo_v$ 
for all $v\in V$. 
We will omit the superscript $\cP$ whenever it is clear from the context. 
Let also $\cW'^\cP(\ov\beta)$, $\cW'^\cP$, $\cW'$ denote the set of 
triples $(v,\al,\beta)$ from $\cW^\cP(\ov\beta)$, $\cW^\cP$, $\cW$, 
respectively, for which $(\al:\beta)=\zo_v$. For every 
$(v,\al,\beta)\in\cW(\ov\beta)$, let $Z(v,\al,\beta,\ov\beta)$ denote the set of triples 
$(w,\gm,\dl)\in\cW(\ov\beta)$ such that $(\al,\beta)$ and $(\gm,\dl)$
cannot be separated in $\rel^{vw}$. Slightly abusing the terminology we 
will also say that $(\al,\beta)$ and $(\gm,\dl)$ cannot be separated in $\cP$. 
Then let $W(v,\al,\beta,\ov\beta)=\{w\in V\mid (w,\gm,\dl)\in Z(v,\al,\beta,\ov\beta)
\text{ for some $\gm,\dl\in\Con(\zA_w)$}\}$.
We will omit mentioning of $\ov\beta$ whenever possible. Sets of the form
$W(v,\al,\beta,\ov\beta)$ will be called \emph{$\ov\beta$-coherent sets},
or just coherent sets if $\ov\beta$ is clear from the context. Also, if
$(\al:\beta)\ne\zo_v$ then the corresponding coherent set is called 
\emph{non-central}. The following statement is an easy corollary of 
Lemma~\ref{lem:delta-alignment}.

\begin{theorem}\label{the:centralizer-alignment}
Let $\cP=(V,\cC)$ be a (2,3)-minimal instance and $(v,\al,\beta)\in\cW$. 
For $w\in W(v,\al,\beta,\ov\beta)$, where $\beta_v=\zo_v$ for $v\in V$, 
let $(w,\gm,\dl)\in\cW$ be such that $(\al,\beta)$ and 
$(\gm,\dl)$ cannot be separated and $\zeta_w=(\gm:\dl)$. Then 
$\cP_{W(v,\al,\beta,\ov\beta)}$ is $\ov\zeta$-aligned.
\end{theorem}

Theorem~\ref{the:centralizer-alignment} relates domains with
congruence intervals that cannot be separated with strands.

\begin{corollary}\label{cor:separation-strand}
Let $\cP=(V,\cC)$ be a (2,3)-minimal instance and $W$ a non-central
coherent set. Then $W$ is a subset of a strand.
\end{corollary}

For technical reasons we will also count the empty set as a
non-central coherent set.

\subsection{Compressed problems}\label{sec:strategies}

In this section we define a way to tighten a block-minimal problem instance 
in such a way that it remains (similar to) block-minimal. More precisely, 
we introduce several properties of a subproblem of a CSP instance 
$\cP$ that are preserved when the problem is restricted in a certain way.

Let $\cP=(V,\cC)$ be a (2,3)-minimal and block-minimal instance over $\cA$.
Recall that for a strand $W\sse V$ by $\cP_{/W}$ we denote the problem
$\cP\fac{\ov\mu_{/W}}$, where $\ov\mu_{/W}=\ov\mu^Y$ and 
$Y=\Razm(\cP)-W$. Let also $\cS_{/W}$ denote
the set of solutions of $\cP_{/W}$. If $W$ is a non-central coherent set, 
the problem $\cP_{/W}$ is defined in the same way. 

\begin{lemma}\label{lem:coherent-block-minimal}
Let $\cP$ be a (2,3)-minimal and block minimal problem. Then for 
every non-central coherent set $W$ the problem $\cP_{/W}$ is 
minimal.
\end{lemma}

\begin{proof}
By Corollary~\ref{cor:separation-strand} there is a strand $U\sse V$
such that $W\sse U$. It now suffices to observe that for every 
solution $\vf\in\cS_{/U}$ of $\cP_{/U}$ the mapping 
$\vf\fac{\ov\mu_{/W}}$ is a solution of $\cP_{/W}$.
\end{proof}

Let $\beta_v\in\Con(\zA_v)$ and let $B_v$ be a $\beta_v$-block, 
$\ov\beta=(\beta_v\mid v\in V)$, $\ov B=(B_v\mid v\in V)$. 
A problem instance $\cP^\dg=(V,\cC^\dg)$, where 
$\ang{\bs,\rel^\dagger}\in\cC^\dg$ if and only if 
$\ang{\bs,\rel}\in\cC$, is said to be \emph{$(\ov\beta,\ov B)$-compressed}
from $\cP$ if the following conditions hold:

\begin{itemize}
\item[(S1)]
For every $\ang{\bs,\rel}\in\cC$ the relation $\rel^\dagger$ is a 
nonempty subalgebra of $\rel\cap\ov B$;
\item[(S2)]
the relations $\rel^{X\dagger}$, where $\rel^{X\dagger}$ is
obtained from $\rel^X$ for $X\sse V$, $|X|\le 2$, form a nonempty 
$(2,3)$-strategy for $\cP^\dg$; 
\item[(S3)]
for every non-central coherent set $W$ the problem 
$\cP^\dg_{/W}=\cP^\dg\fac{\ov\mu_{/W}}$ is minimal;
\item[(S4)]
for every $\ang{\bs,\rel}\in\cC$ the relation $\rel$ is chained 
with respect to $\ov\beta,\ov B$, and the relation $\cS_{/W}$ is 
chained with respect to $\ov\beta,\ov B$ for every non-central
coherent set $W\sse V$;
\item[(S5)]
for every $\ang{\bs,\rel}\in\cC$ the subalgebra $\rel^\dagger$ is
polynomially closed in $\rel$;
\item[(S6)]
for every $\ang{\bs,\rel}\in\cC$ the subalgebra $\rel^\dg$ is
weakly as-closed in $\rel\cap\ov B$.
\end{itemize}

Conditions (S1)--(S3) are the conditions we actually want to maintain when
constructing a compressed instance, and these are the ones that provide 
the desired results. However, to prove that (S1)--(S3) are preserved under 
transformations of compressed instances we also need more technical 
conditions (S4)--(S6).

We now show how we plan to use compressed instances. Let $\cP$ be 
a subdirectly irreducible, (2,3)-minimal, and block-minimal instance, 
$\beta_v=\zo_v$ and $B_v=\zA_v$ for $v\in V$. Then as is easily seen 
the instance $\cP$ itself is $(\ov\beta,\ov B)$-compressed from $\cP$. 
Also, by (S1) a $(\ov\gm,\ov D)$-compressed instance with 
$\gm_v=\zz_v$ for all $v\in V$ gives a solution of $\cP$. Our goal is 
therefore to show that a $(\ov\beta,\ov B)$-compressed instance for any
$\ov\beta$ and an appropriate $\ov B$ can be `reduced', that is, 
transformed to a $(\ov\beta',\ov B')$-compressed instance
for some $\ov\beta'<\ov\beta$. Note that this reduction of instances is 
where the condition $\Razm(\cP)\cap\Centr(\cP)=\eps$ is used. Indeed, 
suppose that $\beta_v=\mu^*_v$ (see Section~\ref{sec:the-algorithm}). 
Then by conditions (S1)--(S6) we only 
have information about solutions to problems of the form 
$\cP\fac{\ov\mu^*}$ or something very close to that. Therefore this 
barrier cannot be penetrated. We  consider two cases.

\medskip
{\sc Case 1.} There are $v\in V$ and $\al\prec\beta_v$ nontrivial on $B_v$, 
$\typ(\al,\beta_v)=\two$. This case is considered in 
Section~\ref{sec:affine-consistency}.

\smallskip
{\sc Case 2.} For all $v\in V$ and $\al\prec\beta_v$ nontrivial on $B_v$, 
$\typ(\al,\beta_v)\in\{\three,\four,\five\}$. This case is considered 
in Section~\ref{sec:non-affine}.

\smallskip
There is also the possibility that $\al\red{\rel^{v\dg}}=
\beta_v\red{\rel^{v\dg}}$ for all $\al\prec\beta_v$. In this case we can 
replace $\beta_v$ with a smaller congruence without violating any of the 
conditions (S1)--(S6).

\section{Proof of Theorem~\ref{the:non-central}: Affine factors}%
\label{sec:affine-consistency}

In this section we consider Case~1 of tightening instances: there is 
$\al\in\Con(\zA_v)$ for some $v\in V$ such that $\al\prec\beta_v$ 
and $\typ(\al,\beta_v)=\two$.

\subsection{Tightening the instance and induced congruences}%
\label{sec:transformation1}

Let $\cP=(V,\cC)$ be a block-minimal instance with subdirectly irreducible 
domains, $\ov\beta=(\beta_v\in\Con(\zA_v)\mid v\in V)$ and 
$\ov B=(B_v\mid B_v \text{ is a $\beta_v$-block, } v\in V)$. 
Let $\cW,\cW'$ denote $\cW^\cP(\ov\beta)$, $\cW'^\cP(\ov\beta)$,
respectively.
Let also $\cP^\dg=(V,\cC^\dg)$ be a $(\ov\beta,\ov B)$-compressed 
instance, and for $C=\ang{\bs,\rel}\in\cC$ there is 
$C^\dagger=\ang{\bs,\rel^\dagger}\in\cC^\dg$. 
We select $v\in V$ and $\al\in\Con(\zA_v)$ with 
$\al\prec\beta_v$, $\typ(\al,\beta_v)=\two$, and 
an $\al$-block $B\in B_v\fac\al$. Note that since $\typ(\al,\beta_v)=\two$, 
$B_v\fac\al$ is a module, and therefore $B$ is as-maximal in this
set. In this section we show how $\cP^\dg$ 
can be transformed to a $(\ov\beta',\ov B')$-compressed instance such
that $\beta'_w\le\beta_w$, $B'_w\sse B_w$ for $w\in V$, and 
$\beta'_v=\al$, $B'_v=B$. Let also $W=W(v,\al,\beta_v,\ov\beta)$, and
let $\cS^\dagger_{/U}$ denote the set of solutions of $\cP^\dg_{/U}$
for a non-central coherent set $U$. We use $\cP^\dg_{/\eps}$,  
$\cS^\dg_{/\eps}$ to denote such a problem and its solution set 
for $U=\eps$. Let also $\cS^\dg_{/U}(B)=\{\vf\in\cS^\dg_{/U}\mid 
\vf(v)\in B\fac{\mu_{/U v}}\}$.

Let $\cP^\ddg=(V,\cC^\ddg)$ be the following instance.
\begin{itemize}
\item[(R1)]
For every $C^\dagger=\ang{\bs,\rel^\dagger}\in\cC^\dg$, the set 
$\rel'^\ddg$ includes 
\begin{itemize}
\item[(a)] 
if $(v,\al,\beta_v)\not\in\cW'$, every $\ba\in\umax(\rel^\dg)$ 
such that $\ba\fac{\ov\mu_{/W}}$ extends to a solution 
$\vf\in\umax(\cS^\dg_{/W}(B))$;\\
\item[(b)] 
if $(v,\al,\beta_v)\in\cW'$, every $\ba\in\umax(\rel^\dagger)$ such that 
$\ba\fac{\ov\mu^\eps}$ extends to a solution 
$\vf\in\umax(\cS^\dg_{/\eps}(B))$. 
\end{itemize}
\item[(R2)]
for every $C^\dagger=\ang{\bs,\rel^\dagger}\in\cC^\dg$, there is 
$C^\ddagger=\ang{\bs,\rel^\ddagger}$, where 
$\rel^\ddagger=\Sgg{\rel}{\rel'^\ddg}$. 
\end{itemize}

The following two statements show how relations $\rel^\ddagger$ 
are related to $\rel^\dagger$. They amount to saying that either 
$\rel^\ddagger$ is (almost) the intersection of $\rel^\dagger$ with 
a block of a congruence of $\rel$, or 
$\umax(\rel^\ddagger)=\umax(\rel^\dagger)$. Recall that for 
congruences $\beta_w$, $w\in V$, and $U\sse V$ by $\ov\beta_U$ we 
denote the collection $(\beta_w)_{w\in U}$.

\begin{lemma}\label{lem:congruence-restriction}
Let $C=\ang{\bs,\rel}\in\cC$, and let 
$\cS^\circ,\cS^{\circ\dagger}$ be the set of solutions of $\cP_{/W}$ 
(respectively, $\cP^\dagger_{/W}$) if $(v,\al,\beta_v)\not\in\cW'$, 
or the set of solutions of $\cP_{/\eps}$ (respectively, 
$\cP^\dg_{/\eps}$) if $(v,\al,\beta_v)\in\cW'$. 
There is a congruence $\tau_C$ 
of $\rel$ satisfying the following conditions.
\begin{itemize}\itemsep0pt
\item[(a)] 
Either $\umax(\rel^\ddagger)=\umax(\rel^\dagger)$, or for a 
$\tau_C$-block $\reli$ it holds $\rel^\ddg=\rel^\dg\cap\reli$. 
\item[(b)] 
Either $\tau_C\red{\rel^\dg}=\ov\beta_\bs\red{\rel^\dg}$, 
or $\rel^\dg\fac{\tau_C}$ is isomorphic to $\rel^{v\dg}\fac\al$. 
Moreover, in the latter case $\tau_C\prec\ov\beta_\bs$.
\end{itemize}
\end{lemma}

If, according to item (b) of the lemma, 
$\tau_C\red{\rel^\dg}=\ov\beta_\bs\red{\rel^\dg}$,
we say that $\tau_C$ is the \emph{full congruence};
if the latter option of item (b) holds we say that $\tau_C$ is a 
\emph{maximal congruence}.

\begin{proof}
If $v\in\bs$ then set $\tau_C$ to be $\ov\beta_\bs\wedge\al$, where 
$\al$ is viewed as a congruence of $\rel^\dagger$, equal to 
$\al\tm\prod_{x\in \bs-\{v\}}\zo_x$. Otherwise consider 
$\relo=\pr_{\bs\cup\{v\}}\cS^\circ$ as a subdirect product of $\zA_v$ and
$\pr_\bs\cS^\circ$. This relation is chained with respect to $\ov\beta,\ov B$ 
by (S4) for $\cP^\dg$ and 
$\pr_{\bs\cap\{v\}}\cS^{\circ\dagger}$ is polynomially 
closed in $\relo$ by (S5) for $\cP^\dg$ and Lemma \ref{lem:poly-closed}(2); 
apply the Congruence 
Lemma~\ref{lem:affine-link} to it. Specifically, consider $\relo\fac\al$ as a 
subdirect product of $\pr_\bs\cS^\circ$ and $\zA_v\fac\al$. If the first option 
of the Congruence Lemma~\ref{lem:affine-link} holds, set 
$\tau_C=\ov\beta_\bs$. If the second option is the case, choose 
$\tau_C$ to be the congruence $\eta$ of $\pr_\bs\cS^\circ$ identified 
in the Congruence Lemma~\ref{lem:affine-link}. Note that in the 
latter case the restriction of $\tau_C$ on $\rel^\dagger$ is nontrivial, 
because tuples from a $\tau_C$-block are related in $\relo$ only to 
elements from one $\al$-block, while the domain of $v$ in $\relo$
spans more than one $\al$-block.

(a) In this case the result follows by the Congruence 
Lemma~\ref{lem:affine-link}.

(b) If $\tau_C\ne\ov\beta_\bs$, by construction 
$\rel^\dg\fac{\tau_C}$ is isomorphic
to $\pr_v\cS^{\circ\dg}\fac\al$, which is isomorphic to $\rel^{v\dg}$. 

To show that $\tau_C\prec\ov\beta_\bs$, as $\beta_w$, $w\in\bs$, 
is the smallest congruence for which $\rel^{w\dg}$ is a subset of a 
$\beta_w$-block, it suffices to prove that for any $\ba,\bb\in\rel^\dg$ and 
such that $\ba\not\eqc{\tau_C}\bb$, $\rel^\dg$ is in a $\gm$-block,
where $\gm=\Cgg\rel{\tau_C\cup\{\ba,\bb\}}$. Consider 
again the relation $\relo$ and let $\rel'=\rel^\dg\fac{\ov\mu^\circ}$,
$\ba'=\ba\fac{\ov\mu^\circ},\bb'=\bb\fac{\ov\mu^\circ}$. Tuples
$\ba,\bb$ can be chosen u-maximal in their $\tau_c$-blocks.
Let also $(\ba',a),(\bb',b)\in\relo$; then $a\not\eqc\al b$ and $a$ can be chosen
u-maximal in its $\al$-block. Since $\al\prec\beta_v$ and $B_v\fac\al$ is a 
module, for any $\al$-block $D\sse B_v$ there is $c\in D$ such that 
$\{a,c\}$ is an $(\al,\beta_w)$-subtrace. 
By Lemma~\ref{lem:collapsing} there is 
a polynomial $f$ of $\pr_{\bs\cup\{v\}}\cS^\circ$ such that 
$f(\ba',a)=(\ba',a)$ and $f(\bb',b)=(\bc',c)$ for some $\bc'\in\rel'$. 
Indeed, we start with any polynomial $g$ that maps $a\fac\al,b\fac\al$
to $a\fac\al,c\fac\al$ and $g(\zA_v)$ is an $(\al,\beta_v)$-minimal set.
Then by Lemma~\ref{lem:collapsing} it can be amended in such a way that 
$g(a)=a$ and $g(\ba')=\ba'$.
Since $a\fac\al b\fac\al$ is 
an affine edge there is also $(\bd,d)\in\Sgg\relo{(\ba',a),(\bc',c)}$
such that $(\ba',a)(\bd',d)$ is a thin affine edge and $d\eqc\al c$.
Since $\relo$ is polynomially closed $(\bd,d)\in\relo$. On the other hand,
as $(\bd,d)\in\Sgg\relo{(\ba',a),(\bc',c)}$, there is a term operation 
$h$ such that $(\bd,d)=h((\ba',a),(\bc',c)$. The polynomial
$h(f(\ba',a),f(x))$ maps $(\ba',a)$ to $(\ba',a)$ and $(\bb',b)$ to
$(\bd,d)$, proving that any two $\tau_C$ blocks of $\rel^\dg$
are $\gm$-related.
\end{proof}

Next we identify variables $w\in V$ for which $\beta'_w$ has to be
different from $\beta_w$. Since $\cP$ is (2,3)-minimal, for every 
$w\in V$ there is $C^w=\ang{(w),\rel^w}\in\cC$. 
For $w\in W$ there are two cases. In the first
case, when $\tau_{C^w}$ is the full congruence, we set
$\beta'_w=\beta_w$. Otherwise $\tau_{C^w}$ is a congruence 
of $\zA_w$ with $\tau_{C^w}\prec\beta_w$ in $\Con(\zA_w)$. Set
$\beta'_w=\tau_{C^w}$. If $\beta'_w\ne\beta_w$ then there 
is a $\beta'_w$-block $B'_w$ such that $b\in B'_w$ whenever 
$(a,b)\in\rel^{vw\dg}$ and $a\in B$. For the remaining variables $w$
we set $B'_w=B_w$.

\begin{lemma}\label{lem:restricted-congruence}
In the notation above
\begin{itemize}\itemsep0pt
\item[(1)]
Let $\gm,\dl\in\Con(\zA_u)$, 
$u\in U=\bs\cap W$ be such that $(u,\gm,\dl)\in\cW$ and 
$(\al,\beta_v),(\gm,\dl)$ cannot be separated from each other. Then 
if $\tau_C$ is a maximal congruence, for any polynomial $f$ of $\rel$, 
$f(\ov\beta_\bs)\sse\tau_C$ if and only if $f(\dl)\sse\gm$. 
If $\gm,\dl$ are considered as congruences of $\rel$,
this condition means
that $(\tau_C,\ov\beta_\bs)$ and $(\gm,\dl)$ cannot be separated.
\item[(2)]
Assuming $\Razm(\cP)\cap\Centr(\cP)=\eps$, if $(v,\al,\beta_v)\in\cW'$,
then for any $w\in\Razm(\cP)$, the interval $(\zz_w,\mu_w)$ can be 
separated from $(\al,\beta_v)$ or the other way round, and therefore 
either $(\zz_w,\mu_w)$ can be separated from every 
$(\tau_C,\ov\beta_\bs)$, where $C\in\cC$ is such that $\tau_C$ 
is a maximal congruence, or the other way round.
\end{itemize}
\end{lemma}

\begin{proof}
(1) Let $\cS^\circ$ be defined as in Lemma~\ref{lem:congruence-restriction}
and $\tau_C$ a maximal congruence. Take a polynomial $f$ of $\rel$. 
Since $\cP$ is a block-minimal instance, the polynomial $f$ can be extended
from a polynomial on $\rel$ to a polynomial of $\cS^\circ$, and, 
in particular, to a polynomial of $\pr_{\bs\cup\{v\}}\cS^\circ$; we keep notation 
$f$ for those polynomials. 
Since $\tau_C$ is maximal, by the Congruence Lemma~\ref{lem:affine-link}
the intervals $(\al,\beta_v)$ and $(\tau_C,\ov\beta_\bs)$ in the congruence
lattices of $\zA_v$ and $\rel$, respectively, cannot be 
separated in $\pr_{\bs\cup\{v\}}\cS^\circ$. 
Therefore $f(\beta_v)\sse\al$ if and only if $f(\ov\beta_\bs)\sse\tau_C$.
Since $(\al,\beta_v)$ and $(\gm,\dl)$ cannot be separated in $\cP$, 
the first inclusion holds if and only if $f(\dl)\sse\gm$, and we infer the 
result. 

(2) Since $(v,\al,\beta_v)\in\cW'$, the centralizer $(\al:\beta_v)=\zo_v$. 
On the other hand, if $w\in\Razm(\cP)$, then
$w\not\in\Centr(\cP)$ and $(\zz_w:\mu_w)\ne\zo_w$.
Therefore $(\al,\beta_v)$ can be separated from  
$(\zz_w,\mu_w)$ or the other way round, as it follows from 
Lemma~\ref{lem:delta-alignment}. 
\end{proof}

Now we are in a position to prove that $\cP^\ddg$ is a 
$(\ov\beta',\ov B')$-compressed instance.

\begin{theorem}\label{the:restricting-strategy}
In the notation above, $\cP^\ddg$ is a 
$(\ov\beta',\ov B')$-compressed instance.
\end{theorem}

\subsection{Conditions (S1), and (S4)--(S6)}\label{sec:S5-1}

We start with conditions (S1), and (S4)--(S6).

Condition (S1) is straightforward by construction, item (R2).
Since $B_v\fac\al$ is a module, and therefore is a nontrivial as-component, 
Lemma~\ref{lem:S7} immediately implies that condition (S4) for $\cP^\ddg$ holds.
Condition (S5) is also fairly straightforward.

\begin{lemma}\label{lem:S5-1}
Condition (S5) for $\cP^\ddg$ holds. That is, for every $\ang{\bs,\rel}\in\cC$
the relation $\rel^\ddagger$ is polynomially closed in $\rel$.
\end{lemma}

\begin{proof}
Let $f$ be a polynomial of $\rel$, and let $\ba,\bb\in\rel$ be tuples 
satisfying the conditions of polynomial closeness. Let $\bc\in\Sg{\ba,f(\bb)}$
be such that $\ba\sqq_{as}\bc$ in $\Sg{\ba,f(\bb)}$. By (S5) for $\cP^\dg$, 
$\bc\in\rel^\dagger$. It suffices to show that $\bc$ is in the same 
$\tau_C$ block as $\ba$. However, this is straightforward, because 
$\ba\eqc{\tau_C}\bb$, and as $f(\ba)=\ba$, we also have 
$\ba\eqc{\tau_C}f(\bb)$. Since $\bc\in\Sg{\ba,f(\bb)}$,
it follows $\bc\eqc{\tau_C}\ba$.
\end{proof}

Finally, condition (S6) also holds.

\begin{lemma}\label{lem:S6-1}
Condition (S6) for $\cP^\ddg$ holds. 
\end{lemma}

\begin{proof}
Let $C=\ang{\bs,\rel}\in\cC$.
By Lemma~\ref{lem:congruence-restriction}(a) either
$\umax(\rel^\ddg)=\umax(\rel^\dg)$, in which case we 
are done, or $\rel^\ddg=\rel^\dg\cap T$,
where $T$ is a $\tau_C$-block. If $w\in\bs-W$, then 
$\umax(\pr_w\rel^\ddg)=\umax(\pr_w\rel^\dg)$ and the property 
of weak as-closeness holds for such variables. Otherwise if $\bs\cap W\ne\eps$,
$\rel^\ddg=\rel^\dg\cap\ov B'$. Moreover for any $\ba\in\rel^\dg$
and any $w,u\in\bs\cap W$ it holds $\ba[w]\in B'_w$ if and only if
$\ba[u]\in B'_u$. Let $a\in\umax(\pr_w\rel^\ddg)\sse\umax(\pr_w\rel^\dg)$
and $b\in\pr_w(\rel\cap\ov B')$ such that $a\sqq_{as}b$ in
$\pr_w(\rel\cap\ov B')$. By (S6) for $\rel^\dg$ there is $\bb\in\rel^\dg$ 
such that $\bb[w]=b$. Then, as we observed $\bb\in\rel^\dg\cap\ov B'=
\rel^\ddg$, as required.
\end{proof}

\subsection{Condition (S2)}\label{sec:S2}

Property (S2) is more difficult to prove. We start with a construction 
similar to what we used before and that we will also use in the proof of (S3).

Let $\mu^\circ_z$ denote $\zz_z$ if $z\in W$ and 
$(v,\al,\beta_v)\not\in\cW'$, and $\mu^\circ_z=\mu^\eps_z$ otherwise.
In other words, $\ov\mu^\circ$ is $\ov\mu_{/W}$ if $(v,\al,\beta_v)\not\in\cW'$
and $\ov\mu^\circ$ is $\ov\mu^\eps$ otherwise. Let $\cS^\circ$ be 
the set of solutions of $\cP^\circ=\cP\fac{\ov\mu^\circ}$. Then for 
$C=\ang{\bs,\rel}\in\cC$ we define $\relo_C$ to be a subalgebra of the product 
$\rel\tm\zA_v\fac\al$ that consists of all tuples $(\bb,c')$, 
$\bb\in\rel$, such that, there is a solution $\vf\in\cS^\circ$ with 
$\bb\in\vf(\bs)$, and $\vf(v)\in c'$. By the block-minimality of $\cP$ 
the relation $\relo_C$ is indeed a subdirect product of $\rel$ and $\zA_v\fac\al$, 
and by (S3) for $\cP^\dg$ we have $\relo_C\cap(\rel^\dg\tm\rel^{v\dg}\fac\al)$
is a subdirect product of $\rel^\dg$ and $\rel^{v\dg}\fac\al)$.
Also, by Lemma~\ref{lem:poly-closed}(2,3) $\relo_C$ is polynomially closed. 

\begin{lemma}\label{lem:S2-1}
Condition (S2) for $\cP^\ddg$ holds. That is, the relations 
$\rel^{X\ddg}$, where $\rel^{X\ddg}$ is obtained 
from $\rel^{X\dg}$ as described in (R1),(R2) for $X\sse V$, $|X|\le 2$, 
form a nonempty $(2,3)$-strategy for $\cP^\ddg$.
\end{lemma}

\begin{proof} 
By (S2) for $\cP^\dg$ the relations $\rel^{X\dg}$, $X\sse V$, $|X|\le2$, 
constitute a $(2,3)$-strategy for $\cP^\dg$.
As $\rel^{xy\ddg}$ is generated by $\rel'^{xy\ddg}$, 
it suffices to show that for any tuple $(a,b)\in\rel'^{xy\ddg}$ 
and any $w\not\in \{x,y\}$ there is $c\in\zA_w$ such 
that  $(a,c)\in\rel^{xw\ddg}, (b,c)\in\rel^{yw\ddg}$. By (R1)
$\rel'^{xw\ddg}\sse\umax(\rel^{xw\dg})$ and so by (S2) for 
$\cP^\dg$ there is $d\in\zA_w$ such that $(a,d)\in\umax(\rel^{xw\dg})$, 
$(b,d)\in\umax(\rel^{yw\dg})$. 

Let $\relo_x=\relo_{C^{xw}},\relo_y=\relo_{C^{yw}}$, as defined 
before Lemma~\ref{lem:S2-1}.
As we observed, $\relo_x$ is a subdirect product of 
$\rel^{xw}\tm\zA_v\fac\al$ and by (S3) for $\cP^\dg$ we have 
$\relo_x\cap(\rel^{xw\dg}\tm\rel^{v\dg}\fac\al)$ is a subdirect product 
of $\rel^{xw\dg}$ and $\rel^{v\dg}\fac\al)$.
For the relation $\relo_y$ similar properties hold. 

Consider the relation 
$$
\rela(x,y,w,v_1,v_2)=\rel^{xy}(x,y)\meet\relo_x(x,w,v_1)\meet\relo_y(y,w,v_2),
$$
and $\rela'=\rela\cap\ov B$ and $\rela^*=\rela\fac{\ov\mu^\circ}$. 
It suffices to show that for some $c\in\rel^{w\dg}$ and $e=B$, such that 
$(a,c)\in\umax(\rel^{xw\dg})$ and $(b,c)\in\umax(\rel^{yw\dg})$ 
it holds $(a,b,c,e,e)\in\rela'$.  Indeed, by the definition of 
$\relo_x,\relo_y$ it means that $(a,c)\in\rel^{xw\ddg}$ 
and $(b,c)\in\rel^{yw\ddg}$. As we 
observed above there is $d\in\rel^{w\dg}$ such that 
$(a,d)\in\rel^{xw\dg}$, $(b,d)\in\rel^{yw\dg}$, and the
triple $(a,b,d)$ extends to a tuple from $\rela'$. Note that as 
$(a,b)\in\umax(\rel^{xy\ddg})$, $d$ can be chosen such that 
$(a,b,d)\in\umax(\pr_{x,y,w}\rela')$. Thus, for some 
$e_1,e_2\in B\fac\al$ we have $\ba=(a,b,d,e_1,e_2)\in\rela'$. 
Since $B_v\fac\al$ is a module and therefore is as-connected,
$\ba\in\umax(\rela')$.
On the other hand, by (R1) there is a solution $\vf$ of 
$\cP^\dg\fac{\ov\mu^\circ}$  such that 
$a\in\vf(x)$, $b\in\vf(y)$, and $\vf(v)\in e$. In other words, there are 
$(a',c')\in\rel^{xw\ddg}$ and $(b',c'')\in\rel^{yw\ddg}$
with $a'\eqc{\mu^\circ_x}a$, $b'\eqc{\mu^\circ_y}b$, and
$c'\eqc{\mu^\circ_w}c''$. This also means that $(a',c',e)\in\relo_x$
and $(b',c'',e)\in\relo_y$. 

By the definition of the congruences $\mu^\circ_z$ and 
Lemma~\ref{lem:restricted-congruence}(2) for every $z\in V$ the
interval $(\al,\beta_v)$ can be separated from $(\zz_z,\mu^\circ_z)$ 
or the other way round. Therefore, 
by Lemma~\ref{lem:collapsing}  there exists an idempotent polynomial $f$ of 
$\rela$ satisfying the following conditions:\\[1mm]
(a) $f$ is $\ov B$-preserving;\\[1mm]
(b) $f(\zA_v\fac\al)$ is an $(\al,\beta_v)$-minimal set;\\[1mm]
(c) $f(\mu^\circ_x\red{B_x})\sse\zz_x$, $f(\mu^\circ_y\red{B_y})\sse\zz_y$, 
$f(\mu^\circ_w\red{B_w})\sse\zz_w$.\\[1mm]
Since $\{e,e_1\}$ is an $(\al,\beta_v)$-subtrace of $\zA_v\fac\al$, 
as $B\fac\al$ is a module, and as $\ba$ can be assumed from
$\umax(\rela'')$, $\rela''=\{\bb\in\rela'\mid \bb[v_1]=e_1,\bb[v_2]=e_2\}$, 
by Lemma~\ref{lem:collapsing} for
$\rela$ the polynomial $f$ can be chosen such that\\[1mm] 
(d) $f(e)=e$, $f(e_1)=e_1$ in coordinate position $v_1$; and\\[1mm]
(e) $f(\ba)=\ba$.\\[1mm]
The appropriate restrictions of $f$ are also polynomials of $\relo_x,\relo_y$.
Therefore applying $f$ to $(a',c',e)$ and $(b',c'',e)$ we get 
$(a,c^*,e)\in\relo_x$, $(b,c^*,e')\in\relo_y$, where $c^*=f(c')=f(c'')$
and $e'=f(e)$ in the coordinate position $v_2$ (and so $f(e)=e$ does 
not have to be true in $v_2$). Thus, $\bb=(a,b,c^*,e,e')\in\rela'$. However,
$(a,c^*),(b,c^*)$ do not necessarily belong to $\rel^{xw\dg},
\rel^{yw\dg}$ respectively. To fix this
let $\bc$ be a tuple in $\Sgg{\rela'}{\ba,\bb}$ such that 
$\ba\bc$ is a thin affine edge and $\bc[v_1]=e$. As is easily seen,
$\bc$ has the form $(a,b,c^\circ,e,e'')$. As, $(a,c')\in\rel^{xw\dg}$,
$(b,c'')\in\rel^{yw\dg}$, and these relations are polynomially closed in 
$\rel^{xw},\rel^{yw}$, respectively, 
$(a,c^\circ)\in\rel^{xw\dg}$, $(b,c^\circ)\in\rel^{yw\dg}$, as well.
Since $(a,b,e,e')\in\umax(\pr_{x,y,v_1,v_2}\rela')$, we may assume 
$\bc\in\umax(\rela')$. Finally, repeating the same argument 
we find a polynomial $g$ of $\rela$ satisfying the conditions (a)--(e) 
with $\bc$ in place of $\ba$ and using the $(\al,\beta_v)$-subtrace
$\{e',e\}$ in coordinate position $v_2$ in place of $\{e_1,e\}$. 
Then we conclude that for some $c^\bullet\in\Sgg{\zA_w}{c^\circ,g(c')}$,
such that $c^\circ c^\bullet$ is a thin affine edge it holds 
$(a,b,c^\bullet,e,e)\in\rela$ and $(a,c^\bullet)\in\rel^{xw\dg}$, 
$(b,c^\bullet)\in\rel^{yw\dg}$. 
\end{proof}

\subsection{Conditions (S3)}\label{sec:S3-1}

In this section we prove that $\cP^\ddg$ satisfies conditions (S3).

As before, let $W=W(v,\al,\beta_v,\ov\beta)$. Recall also that for a coherent set 
$U=W(u,\gm,\dl,\ov\beta)$, $(u,\gm,\dl)\not\in\cW'$
by $\ov\mu_{/U}$ we denote a collection of congruences $\mu'_w$,
$w\in V$ such that $\mu'_w=\mu_w$ if $w\in\Razm(\cP)-U$, and 
$\mu'_w=\zz_w$ otherwise.

\begin{lemma}\label{lem:S3-1}
The instance $\cP^\ddg$ satisfies (S3). That is, for every coherent set 
$U$ the problem $\cP^\ddg_{/U}$ is minimal. More precisely, for every 
$\ang{\bs,\rel^\ddagger}\in\cC^\ddg$, and every $\ba\in\rel^\ddagger$,
there is a solution $\vf\in\cS^\ddagger_{/U}$ such that 
$\vf(\bs)=\ba\fac{\ov\mu_{/U}}$.
\end{lemma}

\begin{proof}
For a coherent set $U$ and a constraint $C=\ang{\bs,\rel^\ddagger}$ it 
suffices only to check that tuples $\ba\in\rel'^\ddagger$
are extendable to solutions of $\cS^\ddagger_{/U}$, because 
$\rel^\ddagger$ is generated by $\rel'^\ddagger$.

For a constraint $C'=\ang{\bs',\rel'}\in\cC$, let $\relo_{C'}$ denote the
relation introduced before Lemma~\ref{lem:S2-1}, and 
$\relo'_{C'}=\relo_{C'}\fac{\ov\mu_{/U}}$.
 
Let $\cC_1\sse\cC$ be the set of all constraints $C'$ such that $\tau_{C'}$
is maximal. Let also $V=\{\vc xn\}$, $v=x_i$, $\bs=(\vc xk)$, and 
$\cC_1=\{\vc C\ell\}$, $C_j=\ang{\bs_j,\rel_j}$. Consider the relation 
$$
\reli(\vc xn,\vc v\ell)=\cS_{/U}(\vc xn)\meet
\bigwedge_{j=1}^\ell \relo'_{C_j}(\bs_j,v_j),
$$
and $\reli'=\reli\cap(\ov B\tm (B_v\fac\al)^\ell)$.  Let $\ba\in\rel'^\ddg$ 
and $\ba'=\ba\fac{\ov\mu_{/U}}$. It suffices to show that for some 
$\bc\in\pr_{x_{k+1}\zd x_n}\cS^\dg_{/U}$ and $e=B$ such that 
$(\ba',\bc)\in\umax(\cS^\dg_{/U})$ it holds $(\ba',\bc,e\zd e)\in\reli'$. 

By construction there is a solution $\vf$ of $\cP^\dg\fac{\ov\mu^\circ}$ 
(recall that this problem is $\cP^\dg_{/\eps}$ if 
$(v,\al,\beta_v)\in\cW'$, and is $\cP^\dg_{/W}$ if 
$(v,\al,\beta_v)\not\in\cW'$) such that 
$\ba\fac{\ov\mu^\circ}=\vf(\bs)$ and $\vf(v)\in e$. Since 
$\ba\fac{\ov\mu^\circ}\in\umax(\rel^\dg\fac{\ov\mu^\circ})$, $\vf$ can be
chosen from $\umax(\cS^{\circ\dg})$. The existence of $\vf$ also 
means that for any $C^*=\ang{\bs^*,\rel^*}\in\cC$ there is 
$\bb_{C^*}\in\rel^{*\ddg}$ such that $\bb_{C^*}\fac{\ov\mu^\circ}=\vf(\bs^*)$.
Again, $\bb_{C^*}$ can be chosen from $\umax(\rel^{*\dg})$. We show 
that there exists a solution $\psi\in\cS^\dg_{/U}$ such that 
$\psi(\bs)=\ba'$ and for every $C^*=\ang{\bs^*,\rel^*}\in\cC_1$ 
it holds 
\begin{equation}
(\psi(\bs^*),\bb'_{C^*})\in\tau_{C^*},\label{equ:congruence-condition}
\end{equation}
where we use $\bb'_{C^*}$ to denote $\bb_{C^*}\fac{\ov\mu_{/U}}$. 
In other words, $\psi\in\cS^\ddg_{/U}$, as required. 
By the definition of $\relo_{C_j}$ there exists $e_j\in B_v\fac\al$
such that $(\bb_{C_j},e_j)\in\relo_{C_j}$, and so $(\bb'_{C_j},e_j)\in\relo'_{C_j}$.

By (S3) for $\cP^\dg$ there is $\sg\in\umax(\cS^\dg_{/U})$ with 
$\sg(\bs)=\ba'$. Choose one for which condition
(\ref{equ:congruence-condition}) is true for a maximal number of 
constraints from $\cC_1$. Suppose that (\ref{equ:congruence-condition}) 
does not hold for $C_j=\ang{\bs^*,\rel^*}\in\cC$. Using the solution $\vf$ of 
$\cP^\dg\fac{\ov\mu^\circ}$ we will construct another solution 
$\sg_0\in\cS^\dg_{/U}$ such that (\ref{equ:congruence-condition})
for $\sg_0$ is true for all constraints it is true for $\sg$, and is also true 
for $C_j$.

By the definition of the congruences $\mu^\circ_z$ and 
Lemma~\ref{lem:restricted-congruence}(2) for every $z\in V$ the
interval $(\al,\beta_v)$ can be separated from $(\zz_z,\mu^\circ_z)$ 
or the other way round. Therefore, 
by Lemma~\ref{lem:collapsing}  there exists an idempotent polynomial $f$ of 
$\reli$ satisfying the following conditions:\\[1mm]
(a) $f$ is $\ov B$-preserving;\\[1mm]
(b) $f(\zA_v\fac\al)$ in the coordinate $v_j$ position of $\reli$ 
is an $(\al,\beta_v)$-minimal set; and\\[1mm]
(c) $f(\mu^\circ_{x_q}\red{B_{x_q}})\sse\zz_{x_q}$ for $q\in[n]$.\\[1mm]
Since $\{e,e_j\}$ is a $(\al,\beta_v)$-subtrace of $\zA_v\fac\al$, 
as $B_v\fac\al$ is a module, and $(\sg,\vc e\ell)\in\umax(\reli')$, 
by Lemma~\ref{lem:collapsing} for
$\reli$ the polynomial $f$ can be chosen such that\\[1mm]
(d) $f(e)=e$, $f(e_j)=e_j$ in coordinate position $v_j$; and \\[1mm]
(e) $f((\sg,\vc e\ell))=(\sg,\vc e\ell)$.\\[1mm]
\indent
The appropriate restrictions of $f$ are also polynomials of $\relo'_{C_q}$
and $\rel'$ for each $q\in[\ell]$ and $C'=\ang{\bs',\rel'}\in\cC$.
By (c) for any $C^\circ=\ang{\bs^\circ,\rel^\circ}, 
C^\bullet=\ang{\bs^\bullet,\rel^\bullet}\in\cC$ we have
$f(\bb'_{C^\circ}[w])=f(\bb'_{C^\bullet}[w])$ for each 
$w\in\bs^\circ\cap\bs^\bullet$. This means that $\sg_0=f(\vf)$ is
properly defined by setting $\sg_0(w)=f(\bb'_{C^\bullet}[w])$ for any $w\in V$ and 
$C^\bullet=\ang{\bs^\bullet,\rel^\bullet}\in\cC$ such that $w\in\bs^\bullet$. 
Also, for any constraint 
$C_q\in\cC_1$ for which (\ref{equ:congruence-condition}) holds for $\sg$,
it also holds for $\sg_0$, as 
$f(\sg(\bs_q))=\sg(\bs_q)\eqc{\tau_{C_q}}\bb'_{C_q}$ implies 
$\sg_0(\bs_q)=f(\bb'_{C_q})\eqc{\tau_{C_q}}f(\sg(\bs_q))
\eqc{\tau_{C_q}}\bb'_{C_q}$ in this case. By (e), $\sg_0(\bs)=\ba'$. 
Finally, $f(e)=e$ in the coordinate position $v_j$ of $\relo$, and so 
$\sg_0(\bs_j)\eqc{\tau_{C_j}}\bb'_{C_j}$,
that is, (\ref{equ:congruence-condition}) holds for $C_j$ as well. 

The mapping $\sg_0$ satisfies many of the desired properties, and it is 
a solution of $\cP_{/U}$ because $\sg_0(\bs^\circ)\in\rel^\circ$ for each 
$C^\circ=\ang{\bs^\circ,\rel^\circ}\in\cC$. However, it is not necessarily a 
solution of $\cP^\dg_{/U}$, and so we need to make one more step. 
To convert $\sg_0$ into a solution of $\cP^\dg_{/U}$ consider
$\bc=(\sg,e\zd e)$ and $\bd=(\sg_0,f_1(e)\zd f_\ell(e))$. Note that the action
of the polynomial $f$ in coordinate positions $v_r$ of $\reli$ may differ,
we reflect it by using subscripts in the tuple $\bd$. In the subalgebra of $\reli$
generated by $\bc,\bd$ take $\bc'=(\psi,e'_1\zd e'_\ell)$ such that $\bc\bc'$ 
is a thin affine edge
and $\bc'[v_j]=e'_j=f_j(e)=e$. For every $C^\circ=\ang{\bs^\circ,\rel^\circ}\in\cC$ 
the relation $\rel^{\circ\dg}$ is polynomially closed in $\rel^\circ$ by (S5).
Since $\sg(\bs^\circ)\psi(\bs^\circ)$ is a thin affine edge in the subalgebra 
generated by $\sg(\bs^\circ),\sg_0(\bs^\circ)$, and $\sg_0(\bs^\circ)$ is the 
image of $\bb'_{C^\circ}\in\rel^{\circ\dg}\fac{\ov\mu_{/U}}$ under $f$, 
we get $\psi(\bs^\circ)\in\rel^{\circ\dg}\fac{\ov\mu_{/U}}$, as well. Thus, 
$\psi$ is a solution of $\cP^\dg_{/U}$. 

Since $\sg(\bs)=\sg_0(\bs)=\ba'$, the same holds for $\psi(\bs)$. Also,
for any constraint $\C_q\in\cC_1$ for which $\sg$ satisfies 
(\ref{equ:congruence-condition}) so does $\sg_0$, and therefore $\psi$.
Finally, by construction $\bc'[v_j]=e$, which means that 
(\ref{equ:congruence-condition}) holds for $C_j$ as well.
A contradiction with the choice of $\sg$.
\end{proof}

\section{Proof of Theorem~\ref{the:non-central}: non-affine factors}%
\label{sec:non-affine}

In this section we consider Case~2 of tightening instances: for every $v\in V$
and every $\al\in\Con(\zA_v)$ with $\al\prec\beta_v$ it holds
$\typ(\al,\beta_v)\ne\two$.

Let $\cP=(V,\cC)$ be a (2,3)-minimal and block-minimal instance with 
subdirectly irreducible 
domains, $\ov\beta=(\beta_v\in\Con(\zA_v)\mid v\in V)$ and 
$\ov B=(B_v\mid B_v \text{ is a $\beta_v$-block, } v\in V)$. 
Let also $\cP^\dg=(V,\cC^\dg)$ be a $(\ov\beta,\ov B)$-compressed 
instance, and for $C=\ang{\bs,\rel}\in\cC$ there is 
$C^\dagger=\ang{\bs,\rel^\dagger}\in\cC^\dg$. 
We select $v\in V$ and $\al\in\Con(\zA_v)$ with 
$\al\prec\beta_v$, $\typ(\al,\beta_v)\ne\two$, and 
an $\al$-block $B\in B_v\fac\al$ such that $B$ is as-maximal in 
$\rel^{v\dg}\fac\al$. By (S6) for $\cP^\dg$ for any 
$C=\ang{\bs,\rel}\in\cC$ with $v\in\bs$, the $\al$-block $B$
is also as-maximal in $\pr_v(\rel\cap \ov B)\fac\al$. In particular, 
it is maximal in $B_v\fac\al=(\rel^v\cap B_v)\fac\al$. We show how $\cP^\dg$ 
can be transformed to a $(\ov\beta',\ov B')$-compressed instance such
that $\beta'_w\le\beta_w$, $B'_w\sse B_w$ for $w\in V$, and 
$\beta'_v=\al$, $B'_v=B$. 

By Lemma~\ref{lem:relative-symmetry} if $\rel^{v\dg}\fac\al$ 
contains a nontrivial as-component,
there is a coherent set associated with the triple $(v,\al,\beta_v)$. 
Let $W=W(v,\al,\beta_v,\ov\beta)$ in this case; note that
$(v,\al,\beta_v)\not\in\cW'$, because $(\al:\beta_v)\ne\zo_v$
by Lemma~\ref{lem:as-type-2}(2). Let also $\cS^\dagger_{/U}$ 
denote the set of solutions of $\cP^\dg_{/U}$ for a coherent set $U$.

\begin{lemma}\label{lem:binary-connections}
If $B_v\fac\al$ contains a nontrivial as-component, then for 
every $w\in W$ there is a congruence 
$\al_w\in\Con(\zA_w)$ with $\al_w<\beta_w$, and such that  
$\rel^{vw\dg}$ is aligned with respect to $(\al,\al_w)$, 
that is, for any $(a_1,a_2),(b_1,b_2)\in\rel^{vw\dg}$, 
$a_1\eqc\al b_1$ if and only if $a_2\eqc{\al_w}b_2$.
\end{lemma}

\begin{proof}
It suffices to show that the link congruences $\lnk_1,\lnk_2$ of 
$\relo=\rel^{vw}$ viewed as a subdirect product of $\zA_v\tm\zA_w$ 
are such that $\beta_v\meet\lnk_1\le\al$ and 
$\beta_w\meet\lnk_2<\beta_w$. Since 
$w\in W$ there are $\gm,\dl\in\Con(\zA_w)$ such that 
$\gm\prec\dl\le\beta_w$ and $(\al,\beta_v)$ and $(\gm,\dl)$ cannot be 
separated. By Lemmas~\ref{lem:as-type-2},\ref{lem:delta-alignment} 
it follows that  
$\beta_v\meet\lnk_1\le\al$ and $\lnk_2\meet\dl\le\gm$. We set 
$\al_w=\beta_w\meet\lnk_2<\beta_w$. 
\end{proof}

Let $\cP^\ddg=(V,\cC^\ddg)$ be constructed as follows.
\begin{itemize}
\item[(R)]
Let $\cP'$ be the problem obtained from $\cP^\dg$ by adding extra 
constraint $\ang{\{v\},B}$.
Let $\cP^\ddg$ be the problem obtained from $\cP'$ by establishing 
$(2,3)$-minimality, and the minimality of  
$\cP^\ddg_{/U}$ for every non-central coherent set $U$.
\end{itemize}
Set $\beta'_v=\al$, $B'_v=B$. Let $Z$ be the set of variables $w$ such that
there is a congruence $\al_w<\beta_w$ such that $\rel^{wv\dg}\fac\al$
is the graph of a mapping $\pi_w:\rel^{w\dg}\to\rel^{v\dg}\fac\al$ 
and $\al_w$ is its kernel. For instance, if $B$ belongs to a nontrivial 
as-component, then $Z=W$. For $w\in U$ set $\beta'_w=\al_w$, 
$B'_w=\pi^{-1}(B)$. 
For the remaining variables $w$ set $\beta'_w=\beta_w$, $B'_w=B_w$. 

\begin{lemma}\label{lem:S5-2}
$\cP^\ddg$ satisfies condition (S5). In other words, for every 
$C=\ang{\bs,\rel}\in\cC$, the relation $\rel^\ddg$ is polynomially closed
in $\rel$.
\end{lemma}

\begin{proof}
Condition (S5) holds for $\cP^\dg$.
The instance $\cP^\ddg$ is obtained from $\cP^\dg$ by adding an 
extra constraint (whose relation is polynomially closed in $\zA_v$)
and establishing various sorts of minimality. This means 
that every $\rel^\ddg$ is obtained through a pp-formula of 
polynomially closed relations. By Lemma~\ref{lem:poly-closed} it is
polynomially closed in $\rel$ as well.
\end{proof}

Condition (S4) follows from Lemma~\ref{lem:S7} by the choice of 
$\beta'_v, B$ and (S6) for $\cP^\dg$.

The following two lemmas show that the constraints of $\cP^\ddg$ 
are not empty. We do it by identifying a set of tuples in every constraint 
relation that withstand the propagation algorithms. We start with constructing
such sets for $(2,3)$-minimality. Set
\begin{eqnarray*}
\relo^x&=&\{a\in \amax(\rel^{x\dg})\mid \text{ there is 
$d\in B$ such that } (d,a)\in\rel^{vx\dg}\},
\end{eqnarray*}

\begin{lemma}\label{lem:2-3}
The collection of sets $\relo^{xy}=\rel^{xy\dg}\cap(\relo^x\tm\relo^y)$, 
$x,y\in V$, is a $(2,3)$-strategy for $\cP'$.
\end{lemma}

\begin{proof}
We need to show that for any $x,y,w\in V$ and $(a,b)\in\relo^{xy}$ there is 
$c\in\rel^{w\dg}$ such that $(a,c)\in\relo^{xw}$, $(b,c)\in\relo^{yw}$. By 
(S2) for $\cP^\dg$ there is $c$ with $(a,c)\in\rel^{xw\dg}$, 
$(b,c)\in\rel^{yw\dg}$. Let $e=B$.
Consider the relation $\relo$ below.
\begin{equation}
\relo'(x,y,w,v) = \rel^{xy\dg}(x,y)\meet\rel^{xw\dg}(x,w)\meet 
\rel^{yw\dg}(y,w)\meet \rel^{wv\dg}\fac\al(w,v),\label{equ:relation-Q1} 
\end{equation}
and $\relo=\pr_{xyv}\relo'$.
As is easily seen, it suffices to show that $(a,b,e)\in\relo$ 
for some $c$. Condition (S2) for $\cP^\dg$ also implies that 
$\ba=(a,b,e')\in\relo$ for some $e'$, and $\ba$ can be
chosen as-maximal in $\relo$. We use the Quasi-2-Decomposition 
Theorem~\ref{the:quasi-2-decomp}.
The tuple $\ba$ indicates that $(a,b)\in\pr_{xy}\relo$. It is also easy 
to see that $(a,e)\in\pr_{xv}\relo$ and $(b,e)\in\pr_{yv}\relo$. 
By Theorem~\ref{the:quasi-2-decomp} $(a,b,e'')\in\relo$ for some
$e''$ with $e\sqq_{as}e''$. If $e$ does not belong to a nontrivial 
as-component of $B_v\fac\al$, then $e''=e$. So, suppose that 
$e$ belongs to a nontrivial as-component $E$ of $B_v\fac\al$.

Let $c\in\rel^{w\dg}$ be such that $(a,b,c,e'')\in\relo'$.
If $w\not\in W$, then by the Congruence Lemma~\ref{lem:affine-link}
$(c,e)\in\rel^{wv\dg}\fac\al$ whenever $c\in\umax(D)$, $D=\{d\in\rel^{w\dg},
(d,e^*)\in\rel^{wv}\fac\al \text{ for some $e^*\in E$}\}$.  
Since $(a,b,e'')\in\amax(\relo)$, element $c$ can be chosen from $\amax(D)$. 
Therefore $(a,b,e)\in\relo$. So, assume that $w\in W$. If $x\in W$ or $y\in W$, 
then $e'=e$. Otherwise as is easily seen, $\rel^{xv\dg}\fac\al\sse\pr_{xv}\relo$, 
$\rel^{yv\dg}\fac\al\sse\pr_{yv}\relo$, and $(\al,\beta_v)$ can be 
separated from any $(\gm_x,\dl_x),(\gm_y,\dl_y)$, where 
$\gm_x\prec\dl_x\le\beta_x$, $\gm_y\prec\dl_y\le\beta_y$,
and $\gm_x,\dl_x\in\Con(\zA_x)$, $\gm_y,\dl_y\in\Con(\zA_y)$,
or the other way round. Consider
\[
\rela(x,y,w,v) = \rel^{xy}(x,y)\meet \rel^{xw}(x,w)\meet 
\rel^{yw}(y,w)\meet \rel^{wv}\fac\al(w,v),
\]
by Lemma~\ref{lem:S7} $\rela$ is chained with respect to 
$\ov\beta,\ov B$. Let $\{e_1,e_2\}\in B_v\fac\al$ be an $(\al,\beta_v)$-subtrace.
By Lemma~\ref{lem:collapsing} there is a $\ov B$-preserving polynomial $f$
of $\rela$ such that $f(e_1)=e_1$, $f(e_2)=e_2$, and $|f(B_x)|=|f(B_y)|=1$. 
Therefore $(\al,\beta_v)$ can be separated from every prime interval 
$\gm\prec\dl\le\beta_x\tm\beta_y$ in $\Con(\rel^{xy})$.
Applying the Congruence Lemma~\ref{lem:affine-link} to $\relo$ we 
obtain $\umax(F)\tm E\sse\relo$, where $F=\{(d_1,d_2)\mid 
(d_1,d_2,e^*)\in\relo \text{ for some $e^*\in E$}\}$. 
In particular, $(a,b,e)\in\relo$.
\end{proof}

Let $\cQ=\{\relo^x\mid x\in V\}$. We say that a tuple 
$\ba\in\prod_{i=1}^\ell\zA_{v_i}$, $\vc v\ell\in V$, is \emph{$\cQ$-compatible}
if $\ba[v_i]\in\relo^{v_i}$ for any $i\in[\ell]$.

\begin{lemma}\label{lem:Q-compatibility}
Let $C=\ang{\bs,\rel}\in\cC$. Then for any non-central coherent set $U$ 
and any $\cQ$-compatible tuple $\ba\in\amax(\rel^\dg)$ 
there is a $\cQ$-compatible solution $\vf\in\cS^\dg_{/U}$ such that 
$\vf(\bs)=\ba\fac{\ov\mu_{/U}}$.
\end{lemma}

\begin{proof}
The proof of this lemma follows the same lines as the proof of 
Lemma~\ref{lem:2-3}. We show by induction that for every $I$, 
$\bs\sse I\sse V$, there is $\psi\in\pr_I\cS^\dg_{/U}$ such 
that $\ba'=\psi(\bs)$, where $\ba'=\ba\fac{\ov\mu_{/U}}$ and 
$\psi(w)\in\relo^x$ for all $w\in I$.
The base case, $I=\bs$ is given by (S3) for $\cP^\dg$.

Suppose the claim is proved for some $I$, $\bs\sse I\sse V$, and 
$w\in V-I$. Let also $\psi\in\amax(\pr_I\cS^\dg_{/U})$ be a partial
solution for this set, $I=\{\vc xk\}$, and $I'=I\cup\{w\}$. Let $e=B$.
Consider the following relation
\begin{equation}
\relo'(\vc xk,w,v)=\pr_{I'}\cS^\dg_{/U}(\vc xk,w)
\meet\rel^{wv\dg}\fac\al(w,v),\label{equ:relation-Q2}
\end{equation}
and $\relo=\pr_{I\cup\{v\}}\relo'$. As is easily seen, it suffices 
to show that $(\psi,e)\in\relo$. Firstly, $\psi\in\pr_I\relo$ by the 
induction hypothesis, as any value
of $w$ can be extended to a pair from $\rel^{wv\dg}$.
For $i\in[k]$, as $(\psi(x_i),e)\in\rel^{x_iv\dg}\fac{\mu_{/Ux_i}\tm\al}$, 
we have $(\psi(x_i),b)\in\rel^{x_iv\dg}\fac{\mu_{/Ux_i}}$ for some $b\in B$.
By (S3) for $\cP^\dg$ this pair can be extended to a solution 
from $\cS^\dg_{/U}$. This implies $(\psi(x_i),e)\in\pr_{x_iv}\relo$. 
By the Quasi-2-Decomposition Theorem~\ref{the:quasi-2-decomp} 
$(\psi,e')\in\relo$ for 
some $e'\in\as(e)$ in $B_v\fac\al$. If $e$ is in a trivial as-component 
of $B_v\fac\al$, we obtain $e'=e$. So, suppose that 
$e$ belongs to a nontrivial as-component $E$ of $B_v\fac\al$.

As $\psi$ is as-maximal, there is as-maximal $\vf=(\psi,e')\in\relo$. 
If $w\not\in W$,  by the Congruence Lemma~\ref{lem:affine-link}
$(c,e)\in\rel^{wv\dg}\fac\al$ whenever $c\in\umax(D)$, $c$ satisfies 
the conditions of (\ref{equ:relation-Q2}) and $D=\{d\in\rel^{w\dg},
(d,e^*)\in\rel^{wv\dg}\fac\al \text{ for some $e^*\in E$}\}$.  
Since $\vf\in\umax(\relo)$, element $c$ can be chosen from $D$. 
Therefore $(\psi,e)\in\relo$. So, assume that $w\in W$. If $I\cap W\ne\eps$, 
then $e'=e$. Otherwise as is easily seen, 
$\rel^{x_iv\dg}\fac{\mu_{/Ux_i}\tm\al}\sse\pr_{x_iv}\relo$ and 
$\al\prec\beta_v$ can be separated from any $\gm\prec\dl\le\beta_{x_i}$, 
where $\gm,\dl\in\Con(\zA_{x_i})$, $i\in[k]$. Consider
\[
\rela(\vc xk,w,v) = \pr_{I'}\cS_{/U}(\vc xk,w)
\meet \rel^{wv}\fac\al(w,v),
\]
by Lemma~\ref{lem:S7} $\rela$ is chained with respect to 
$\ov\beta,\ov B$. Similar to the proof of Lemma~\ref{lem:2-3}, let 
$\{e_1,e_2\}\in B_v\fac\al$ be an $(\al,\beta_v)$-subtrace.
By Lemma~\ref{lem:collapsing} there is a 
polynomial $f$ of $\rela$ such that $f(e_1)=e_1$, $f(e_2)=e_2$, and 
$|f(B_{x_i})|=1$ for $i\in[k]$. Therefore $(\al,\beta_v)$ can be 
separated from every prime interval $\gm\prec\dl\le\ov\beta_I$ in 
$\Con(\pr_I\cS_{/U})$. 
Applying the Congruence Lemma~\ref{lem:affine-link} to $\rela$ and $\relo$ we 
obtain $\umax(F)\tm E\sse\relo$, where $F=\{\chi\in\pr_I\cS^\dg_{/U}\mid 
(\chi,e^*)\in\relo \text{ for some $e^*\in E$}\}$. 
In particular, $(\psi,e)\in\relo$.
\end{proof}

Conditions (S2), (S3) hold for $\cP^\ddg$ by construction and 
$\cP^\ddg$ does not contain empty constraint relations by 
Lemmas~\ref{lem:2-3} and~\ref{lem:Q-compatibility}, implying (S1).

Finally, we verify condition (S6).

\begin{lemma}\label{lem:S6-2}
Condition (S6) for $\cP^\ddg$ holds. 
\end{lemma}

\begin{proof}
Similar to the sets $\relo^x$ above we introduce 
\[
\reli^x=\{a\in \rel^{x\dg}\mid \text{ there is 
$d\in B$ such that } (a,d)\in\rel^{xv\dg}\}.
\]
Pick $C=\ang{\bs,\rel}\in\cC$. We make use of the following property of 
$\rel^\dg$: for any $w,u\in\bs\cap Z$ and any $\ba\in\rel^\dg$, if
$\ba[w]\in B'_w$ then $\ba[u]\in B'_u$. 
We prove the claim in three steps. First, we will show that for every 
$C=\ang{\bs,\rel}\in\cC$ the relation $\rel''=\rel^\dg\cap\prod_{x\in\bs}\reli^x$ 
is as-closed (not weakly as-closed!) in $\rel'=\rel^\dg\cap\ov B'$. 
Second, we use Lemma~\ref{lem:poly-closed} to conclude that 
$\rel^\ddg$ is as-closed in $\rel''$. Third, we conclude that this implies 
that $\rel^\ddg$ is weakly as-closed in $\rel\cap\ov B'$.

For the first step, note that it suffices to show that $\reli^x$ is as-closed 
in $\pr_x\rel\cap B'_x$ for $x\in\bs$. Depending on the case of the Congruence 
Lemma~\ref{lem:affine-link} that holds for $\rel^{xv\dg}\fac\al$ and 
whether or not $B$ belongs to a nontrivial as-component $E$, either 
$\umax(\reli^x)=\umax(\reli'^x)$, where 
\[
\reli'^x=\{a\in \rel^{x\dg}\mid \text{ there is 
$d\in B_v$ such that } d\fac\al\in E, \text{ and }(a,d)\in\rel^{xv\dg}\},
\]
or $\reli^x=B'_x$. In both cases the claim holds.

The second step is immediate by Lemma~\ref{lem:poly-closed}. For
the third step, if $a\in\umax(\pr_w\rel^\ddg)\sse\umax(\pr_w(\rel\cap\ov B'))$
and $b\in\pr_w(\rel\cap\ov B')$ are such that $a\sqq_{as} b$ for 
$w\in\bs$, then let $\ba\in\rel^\ddg$ with $a=\ba[w]$. Since $\ba\in\rel^\dg$,
by (S6) for $\rel^\dg$, $b\in\pr_w\rel^\dg$, and therefore $b=\bb[w]$
for some $\bb\in\rel^\dg$. As $a\sqq_{as}b$, the tuple $\bb$ can be chosen
such that $\ba\sqq_{as}\bb$ in $\rel^\dg$. Moreover, as we observed above,
$\bb\in\rel'$. This means, by the second step, that 
$\bb\in\rel^\ddg$, confirming the claim.
\end{proof}

\bibliographystyle{plain}

\end{document}